\documentclass[prl,floatfix,amsmath,twocolumn,notitlepage,superscriptaddress]{revtex4-2}

\usepackage{tikz} % Package for drawing
\usepackage{amsmath,amsthm,amsfonts}
\usepackage{mathtools}
\usepackage{verbatim}
\usepackage{hyperref}
\usepackage{dsfont}
\usepackage{soul} %for strikeout text
\usepackage{array, makecell}
\usepackage{tabularx}
\usepackage{enumitem}% http://ctan.org/pkg/enumitem
\usepackage{calc}% http://ctan.org/pkg/calc

\theoremstyle{definition}

\newtheorem{lemma}{Lemma}
\newtheorem{proposition}{Proposition}
\newtheorem{theorem}{Theorem}
\newtheorem*{prop}{Proposition}

\newcommand{\msf}{\mathsf}
\newcommand{\mbf}{\mathbf}
\newcommand{\mbb}{\mathbb}
\newcommand{\mc}{\mathcal}

\newcommand{\tr}{\textrm{Tr}}

\newcommand{\ket}[1]{|#1\rangle}
\newcommand{\bra}[1]{\langle #1|}
\newcommand{\op}[2]{|#1\rangle\langle #2|}

\renewcommand{\ol}{\overline}

\newcommand{\stab}{\text{stab}}

\definecolor{cool_green}{rgb}{0.0, 0.5, 0.0}
\definecolor{new_pink}{rgb}{1, 0.0, 1}

\begin{document}

\title{Multipartite Nonlocality in Clifford Networks}

\author{Amanda Gatto Lamas}\email{amandag6@illinois.edu}
\affiliation{Department of Physics, University of Illinois at Urbana-Champaign, Urbana, IL 61801, USA}

\author{Eric Chitambar}\email{echitamb@illinois.edu}
\affiliation{Department of Electrical and Computer Engineering, Coordinated Science Laboratory, University of Illinois at Urbana-Champaign,  Urbana,  IL 61801, USA}

	\date{\today}
	
\begin{abstract}

We adopt a resource-theoretic framework to classify different types of quantum network nonlocality in terms of operational constraints placed on the network.  One type of constraint limits the parties to perform local Clifford gates on pure stabilizer states, and we show that quantum network nonlocality cannot emerge in this setting.  Yet, if the constraint is relaxed to allow for mixed stabilizer states, then network nonlocality can indeed be obtained.  We additionally show that bipartite entanglement is sufficient for generating all forms of quantum network nonlocality when allowing for post-selection, a property analogous to the universality of bipartite entanglement for generating all forms of multipartite entangled states.

%In an $N$-party quantum network, non-classical correlations can arise when locally measuring across $k$-party entangled states ($k<N$) that are independently distributed.  Such correlations are said to embody quantum network nonlocality.  In this paper, we adopt a resource-theoretic framework to classify different types of quantum nonlocality in terms of operational constraints placed on a network.  One type of constraint limits the parties to perform local Clifford gates on pure stabilizer states, and we show that quantum network nonlocality cannot emerge in this setting.  Yet, if the constraint is relaxed to allow for mixed stabilizer states, then network nonlocality can indeed be obtained. We additionally show that bipartite entanglement (i.e. $k=2$) is sufficient for generating all forms of quantum network nonlocality when allowing for post-selection, a property analogous to the universality of bipartite entanglement for generating all forms of multipartite entangled states.

%In this paper, we show that bipartite entanglement (i.e. $k=2$) is sufficient for generating all forms of quantum network nonlocality when allowing for post-selection, a feature that might not arise when using more general non-signaling devices. We then demonstrate that network nonlocality fails to emerge in quantum networks limited strictly to pure stabilizer states and Clifford gates. 
%This work sheds new light on the distinction between non-genuine and genuine network nonlocality by showing that all known approaches for generating the latter require some form of non-Clifford operation.

\end{abstract}

\maketitle

%A fundamental problem in quantum information science is to understand which quantum processes can be efficiently simulated by classical means and which cannot.  This question can look rather different for different types of tasks.  For instance, in quantum computing, much effort is devoted to discovering new classical algorithms that can closely reproduce the performance of their quantum counterpart.  On the other hand, in the study of Bell nonlocality one tries to construct local hidden variable (LHV) models that classically reproduce the local measurement statistics on entangled states.  In both cases, there are limits on how much can be classically simulated, but the exact limitations are currently unknown.

%In this paper, we identify a link between the classical simulation of quantum algorithms and the existence of multipartite LHV models.  

The Gottesman-Knill Theorem is a classic result that enables a wide class of quantum algorithms to be efficiently simulated \cite{Gottesman-1998a, Aaronson-2004a, Anders-2006a}.  It says that circuits constructed using only gates belonging to the Clifford group can be efficiently simulated on a classical computer using the stabilizer formalism.  It is interesting to investigate whether Clifford quantum information processing has similar limitations for other quantum information tasks and phenomena, such as the emergence of quantum nonlocality.  It has already been shown \cite{Pusey2010, Tessier2005, Howard2012, Howard2015} that such limitations do indeed exist in specific scenarios, but the role of Clifford operations in the emergence of network nonlocality has remained relatively unexplored.  %\amanda{This work finds that all multipartite correlations generated on networks restricted to stabilizer states and local Clifford gates can be simulated using a classical randomness on the same network.}%  Interestingly, there are subtleties to this conclusion that touch on the distinction between so-called genuine and non-genuine network nonlocality, as we will discuss below. 

The standard Bell nonlocality scenario consists of multiple parties having access to some globally-shared entangled state, and they each select different measurements to perform on their respective subsystems \cite{Brunner-2014a}.  In the network setting, the globally-shared entanglement is replaced by independent sources of entanglement that get distributed according to the structure of the network \cite{Branciard-2010a, Bancal-2011a, Branciard-2012a, Gisin-2019a, Tavakoli-2022a}. 
Examples of nonlocality in networks have recently been found that appear to be fundamentally different than the nonlocality emerging in traditional Bell scenarios \cite{Renou-2019b, Tavakoli-2021a, Renou-2022a, Renou-2022b, Pozas-Kerstjens-2022a, Pozas-Kerstjens-2022b}, although how to best articulate this difference is unclear.

To shed light on this question, we begin here by sketching a ``top-down'' framework for the general study of quantum nonlocality (or ``nonclassicality''), inspired by the philosophy of quantum resource theories \cite{Horodecki-2012a, Chitambar-2019a}.  In this framework, different classes of nonlocality naturally emerge on a quantum network after placing different restrictions on the type of states that are ``free'' to distribute across the network, as well as the type of local operations that the parties are free to perform.  This allows us to view different notions of nonlocality proposed in the literature under a common conceptual lens.  As our main result, we show later in this letter that quantum nonlocality can never be realized whenever the parties are restricted to Clifford operations and the free states are stabilizer states. 

\textit{An operational framework for network nonlocality.}  
\begin{figure}[t]
\centering
\includegraphics[scale=0.69]{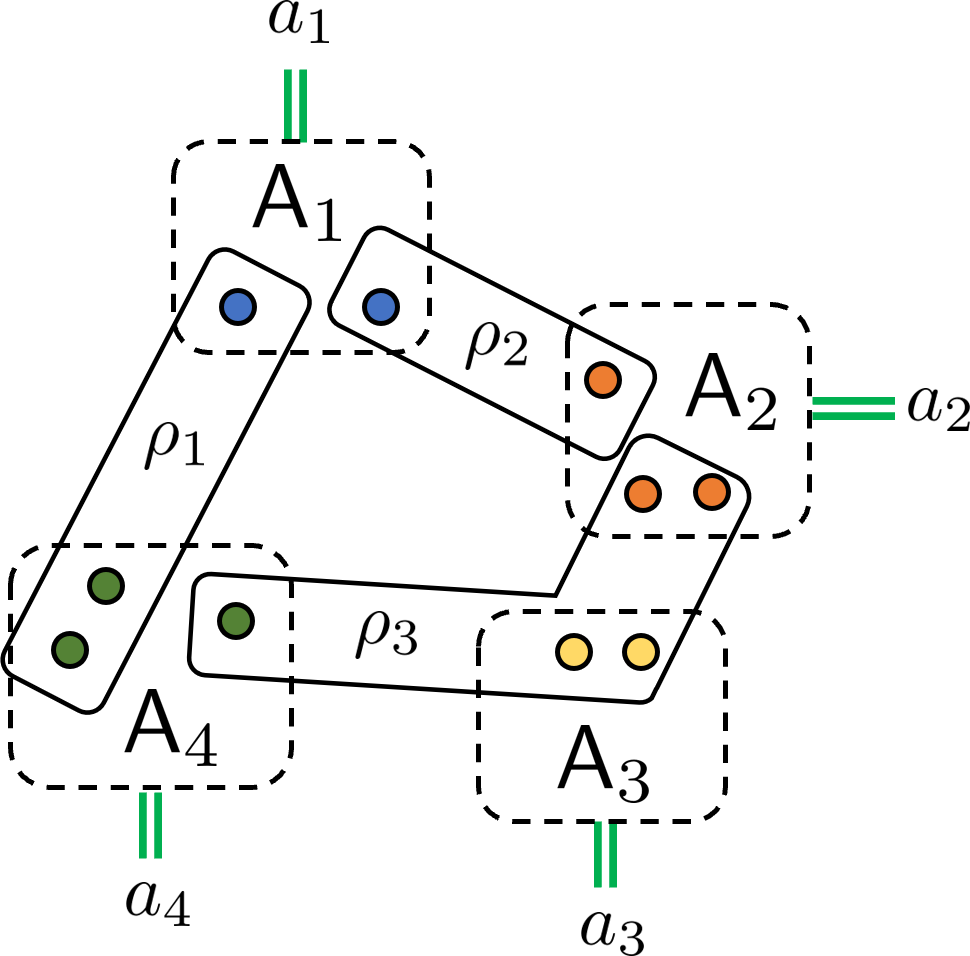}
\caption{A 4-party quantum network partitioned into disjoint hyperedges.  Each node represents a different quantum system and each hyperedge represents a quantum state shared among the constituent nodes.  Each party $\msf{A}_i$ has local control of certain systems.}
\label{Fig:hypergraph}
\end{figure}
Consider an $N$-party quantum network whose structure is defined by a hypergraph $\mc{G}=(V,E)$ with a disjoint hyperedge set $E\subset 2^{V}$.  Each vertex $v\in V$ represents a quantum system, and each hyperedge $e\subset V$ represents an independent quantum source that generates a joint state $\rho_e$ for all $|e|$ systems in $e$.  The network structure also specifies which quantum systems are received locally by each party.  We let $A_i\subset E$ denote the collection of hyperedges connected to party $\msf{A}_i$ (see Fig. \ref{Fig:hypergraph}). %For $i=1,\cdots,N$, let $A_i\subset V$ denote the systems under the local control of party $\msf{A}_i$.  
In any protocol for generating network correlations, each party $\msf{A}_i$ applies a quantum channel $\mc{E}_i$ jointly across all its received systems, thereby connecting the previously disjoint hyperedges.  The parties then measure their systems in the computational basis, and the collective output is the $N$-tuple $(a_1,\cdots,a_N)$ of measured values with probability distribution $p(a_1,\cdots,a_N)$.  In summary, every network correlation we consider can be generated using the following three-step procedure:
\begin{enumerate}
\item[(i)] For each hyperedge $e\in E$ in the network, a multipartite state $\rho_e$ is distributed across the network; 
\item[(ii)]  Each party $\msf{A}_i$ performs a local processing channel $\mc{E}_i$ on the systems under its control; 
\item[(iii)] Each party $\msf{A}_i$ measures its post-processed system in the computational basis and outputs the outcome $a_i$.
\end{enumerate}
A distribution $p(a_1,\cdots,a_N)$ built in this way will be called a \textit{quantum network correlation}.

%Our interest is an understanding the types of distributions $p(\mbf{b}_1,\cdots\mbf{b}_N)$ can be generated on a given network $\mc{G}$.  
Since this framework is designed to study quantum nonlocality, one may want to generalize the setup and grant each party $\msf{A}_i$ an input variable $x_i$ that controls the local processing performed.  In this case, the generated correlation would be a conditional distribution $p(a_1,\cdots,a_N|x_1,\cdots,x_N)$, which is typically the object of consideration in the study of Bell nonlocality.  However, in the network scenario, the distinction between correlations with inputs versus correlations without inputs is superficial.  This is due to the work of Fritz \cite{Fritz-2012a}, who showed how every correlation generated on some network with inputs can be equivalently generated on an enlarged network without inputs; essentially the local input variable $x_i$ of party $\msf{A}_i$ becomes an independent variable shared between $\msf{A}_i$ and some new party on the enlarged network.  Therefore, without loss of generality, we restrict to correlations with no inputs.

With this model in place, we can now identify different classes of quantum correlations by imposing different constraints on the distributed states $\rho_e$ and the types of local maps $\mc{E}_i$.  We begin with the set of local correlations defined with respect to a given network $\mc{G}=(V,E)$.  A network local correlation is any distribution that can be generated using a shared classical variable $\lambda_e$ for each hyperedge $e$.  If $p(\lambda_e)$ denotes the probability distribution over variable $\lambda_e$, then every local correlation has the form
\begin{equation}
p(a_1,\cdots,a_N)=\sum_{\vec{\lambda}}p(\vec{\lambda})\prod_{i=1}^Np(a_i|\cup_{e\in A_i}\lambda_e),\label{Eq:network-local}
\end{equation}
where the sum is over each sequence $\vec{\lambda}=(\lambda_e)_e$ of the $|E|$ independent random variables, $p(\vec{\lambda})=\prod_{e\in E}p(\lambda_e)$, and $p(a_i|\cup_{e\in A_i}\lambda_e)$ is a local classical channel of party $\msf{A}_i$.  If a quantum network correlation does not have this form, then it is called \textit{quantum network nonlocal}.  We claim that the set of local correlations is precisely those that can be generated via steps (i)--(iii) above under the constraint that each $\mc{E}_i$ satisfies the condition
\begin{equation}
\label{Eq:classical-simulatble}
    \Delta\circ\mc{E}_i = \Delta\circ\mc{E}_i\circ\Delta,
\end{equation}
where  $\Delta(\cdot)=\sum_{x}\op{x}{x}(\cdot)\op{x}{x}$ is the completely dephasing map in the computational basis.  This condition has previously been called \textit{resource non-activating} \cite{Liu-2017a}, but in this work we will say that any completely-positive trace-preserving (CPTP) map satisfying Eq. \eqref{Eq:classical-simulatble} is \textit{classically simulatable}.  Indeed, any correlation $p(a_1,\cdots,a_N)$ generated using classically simultable maps $\mc{E}_i$ can be simulated on the same network $\mc{G}$ using purely classical resources.  Since the parties measure their qubits in the computational basis after applying $\otimes_i\mc{E}_i$ in step (ii), Eq. \eqref{Eq:classical-simulatble} says that the parties could equivalently also dephase the shared state $\otimes_e\rho_e$ prior to step (ii), thereby converting each $\rho_e$ into a classically correlated state $\hat{\rho}_e=\sum_{\mbf{x}}p(\mbf{x})\op{\mbf{x}}{\mbf{x}}$, where $\ket{\mbf{x}}=\ket{x_1}\cdots\ket{x_{|e|}}$.  Conversely, any network correlation $p(a_1,\cdots,a_N)$ generated using classical shared randomness can also be produced using classically correlated states and CPTP maps of the form $\mc{E}_i(\cdot)=\sum_{a_i}\op{a_i}{a_i} \tr[\Pi_{a_i} (\cdot)]$, where
\[\Pi_{a_i}=\sum_{e\in A_i}\sum_{\lambda_e} p(a_i|\cup\lambda_e) \bigotimes_{e\in A_i}\op{\lambda_e}{\lambda_e}.\]
\begin{table}[t]
\centering
\begin{tabular}{|p{10em} || p{14em} |} 
\hline
{Operational constraint} & {Inaccessible type of network correlation } \\ 
 \hline\hline
\mbox{ $\mc{E}_i$ are classically } simultable & Network nonlocality  \\ \hline
$\mc{E}_i$ are local quantum wiring maps & Genuine network nonlocality \cite{Brunner_2022} \\\hline
\mbox{ $\mc{E}_i$ are classically } simultable for at least one party in a given subgraph & Full network nonlocality for the given subgraph \cite{Pozas-Kerstjens-2022a} \\ \hline
 \mbox{ $\mc{E}_i$ are Clifford gates} and 
 $\rho_e$ are pure \mbox{stabilizer states} & Network nonlocality\\
 \hline
\end{tabular}
\caption{Different classes of classical correlations emerging on a quantum network based on the operational constraints.}
\label{Tab-correlations}
\end{table}

Another type of constraint limits each $\mc{E}_i$ to be a local quantum wiring map for party $\msf{A}_i$ \cite{Gallego-2012a, deVicente-2014a, Brunner_2022}, which can be understood as a local operations and classical communication (LOCC) transformation performed on the separated systems that $\msf{A}_i$ receives from different sources.  Under this constraint, $\msf{A}_i$ is prohibited from performing entangling measurements across the different states it receives, such as in entanglement swapping.  Nevertheless, these operations are still strong enough to generate nonlocal correlations when seeded with at least one entangled state $\rho_e$ on the network.  All ``standard'' Bell tests - such as the celebrated violations of the Clauser-Horne-Shimony-Holt (CHSH) inequality \cite{Hensen-2015a, Shalm-2015a, Giustina-2015a} - can be implemented under the restriction of local quantum wiring maps.  In contrast, network correlations not producible by local quantum wiring maps can be defined as possessing \textit{genuine network nonlocality} \cite{Brunner_2022}, much in the same way that quantum states not producible by LOCC are defined as possessing entanglement \cite{Horodecki-2009a}.  This definition is motivated by the fact that to produce genuine quantum network nonlocality, the parties must be able to truly leverage the network structure and stitch together the different quantum sources through non-classical local interactions.

%Nonlocal correlations generated using local quantum wiring maps have been identified as \textit{non-genuine network nonlocal} since they involve stitching together standard (non-network) Bell tests by classical processing.   

To provide a complete account of the different notions of network nonlocality proposed in the literature, we consider one final type of correlation.  For a fixed subgraph of a given network, one could require that at least one of the constituent parties performs a classically simulatable operation $\mc{E}_i$.  Nonlocal correlations that cannot generated under this restriction are said to possess \textit{full network nonlocality} with respect to the particular subgraph \cite{Pozas-Kerstjens-2022a}.

We now enjoin the main result of this letter to the picture.  Let $p(a_1,\cdots,a_n)$ be any quantum network correlation built under the constraint that all $\rho_e$ are pure stabilizer states and the $\mc{E}_i$ are Clifford gates.  Then $p(a_1,\cdots,a_n)$ is a network local distribution.  The rest of this letter will be devoted to explaining this result and its proof in more detail; technical steps are postponed to the Supplemental Material.  Table \ref{Tab-correlations} summarizes the different types of network correlations in the context of the operational framework outlined above. 

\textit{$k$-network nonlocality.}  Before specializing to Clifford networks, let us first sharpen the notion of local correlations to reflect even better the structure of the network.  Let us say that a $k$-network is any hypergraph with disjoint edge set in which every quantum source is connected to no more than $k$ parties; i.e. $|\{i\;|\;e\cap A_i\not=\emptyset\}|\leq k$ for all $e\in E$.  An $N$-party distribution $p(a_1,\cdots,a_N)$ that can be built using classically simulatable operations on some $k$-network will be called \textit{$k$-network local}; otherwise it will be called \textit{quantum $k$-network nonlocal}. 

It turns out that every instance of quantum $k$-network nonlocality can be obtained from a quantum $2$-network nonlocal correlation through post-selection.  
\begin{proposition}
\label{Prop:k-to-2}
Suppose that $p(a_1,\cdots,a_N)$ is a quantum $k$-network nonlocal correlation for $N$ parties.  Then there exists an $(N+r)$-party correlation $\hat{p}(a_1,\cdots,a_N, c_{1},\cdots,c_r)$ that is quantum $2$-network nonlocal and satisfies
\begin{equation}
p(a_1,\cdots,a_N)=\hat{p}(a_1,\cdots,a_N| c_{1}=0,\cdots,c_r=0).
\end{equation}
In other words, conditioned on the new parties $\msf{C}_{1},\cdots,\msf{C}_r$ having the all-zero output, the other $N$ parties reproduce the original distribution $p$. 
\end{proposition}

%A proof of this proposition is provided in the Supplemental Material at \textbf{[URL will be provided by publisher]}. 
To prove this proposition one replaces every $k$-element hyperedge $e$ with $k$ bipartite edges and then uses bipartite entanglement distributed on these edges to teleport the original state $\rho_e$.  We remark that our reduction to $k$-network nonlocality from $2$-network nonlocality via post-selection is specific to quantum networks, as it relies on quantum teleportation.  It is an interesting question whether such a reduction can be found for general non-signaling network correlations \cite{Gisin-2020a, Bierhorst-2021a} or whether this is a uniquely quantum feature. %\amanda{It is important to note that this post-selection is only used as a tool to show that any $k$-network nonlocality could be achieved from networks with only bipartite correlations under appropriate circumstances. However, post-selection is not used to obtain the correlations that our main result shows to be exclusively local, as delineated in the next section.}

\textit{Clifford Networks.}  Let us now turn to the notion of Clifford quantum networks, which can be understood as distributed Clifford circuits \cite{Gottesman-1998a}.  We begin by recalling the definitions of stabilizers and stabilizer operations \cite{Gottesman-1997a, Gottesman-1998a, Nielsen-2000a}.  Let $\mc{P}_n$ denote the $n$-qubit Pauli group of operators, $\mc{P}_n=\{\pm 1,\pm i\}\times\{I,X,Y,Z\}^{\otimes n}$.  Expressions like $X_2$ express Pauli-$X$ applied to qubit $2$ and the identity applied to all other qubits.  The $n$-qubit Clifford group consists of all unitaries that, up to an overall phase, leave $\mc{P}_n$ invariant under conjugation,
\begin{equation}
    \mc{C}_n=\{U\;|\; UgU^\dagger\in\mc{P}_n\quad \forall g \in \mc{P}_n\}/U(1). 
\end{equation}
The set of $n$-qubit stabilizer states is defined as
\begin{equation}
    \mc{S}_n=\{U\ket{0}\;|\; U\in\mc{C}_n\}.
\end{equation}
%along with all density matrices formed by taking convex combinations of elements from $\mc{S}_n$.  
Every stabilizer pure state $\ket{\varphi}\in\mc{S}_n$ can be uniquely identified as the $+1$ eigenstate of $n$ independent commuting elements from $\mc{P}_n$.  These operators generate a group, called a stabilizer group, that we denote by $\stab(\ket{\varphi})$.

Consider now a $k$-network $\mc{G}=(V,E)$ in which each vertex $v\in V$ represents a qubit system.  The hyperedges $E$ again form a disjoint partitioning of the vertex set and each $e\in E$ represents a multi-qubit state.  Suppose that party $\msf{A}_i$ receives a total of $n_i$ qubits from the various sources. Using the three-step framework introduced above, we consider correlations formed under the following restrictions:
\begin{enumerate}
    \item[(ic)] For each hyperedge $e\in E$, a pure stabilizer state $\ket{\varphi_e}$ is distributed;
    \item[(iic)] Each party $\msf{A}_i$ introduces $m_i$ ancilla qubits, each initialized in state $\ket{0}$, and performs a Clifford unitary gate $V_i$  on all the $n_i+m_i$ qubits held locally. 
\end{enumerate}
Like before, step (iii) involves a measurement of each qubit in the computational basis.  This generates a classical output for each party that is a sequence of bits $\mbf{b}_i=(b_1,b_2,\cdots,b_{n_i+m_i})\in\mbb{Z}_2^{n_i+m_i}$, one for each qubit used by party $\msf{A}_i$ in the protocol.  Letting $n=\sum_{i=1}^Nn_i$ and $m=\sum_{i=1}^N m_i$, the joint probability distribution for all measurements is then given by
\begin{align}
    p(\mbf{b}_1,\cdots,\mbf{b}_N)&=p(b_1,b_2,\cdots,b_{n+m})\notag\\
    &=\frac{1}{2^{n+m}}\bra{\omega}\bigotimes_{i=1}^{n+m}(\mbb{I}+(-1)^{b_i}Z_i)\ket{\omega}.
\end{align}
Every distribution $p(\mbf{b}_1,\cdots,\mbf{b}_N)$ having this form will be called a \textit{Clifford $k$-network correlation}.  As an example, a triangle Clifford network is depicted in Fig. \ref{Fig:network-pic}.  Its equivalent representation as a distributed quantum circuit is also shown.  

Observe that the Clifford constraint and the classical simulatable constraint are inequivalent; i.e., there are Clifford gates $U$ (such as Hadamard) whose corresponding CPTP map $\mc{U}(\cdot)=U(\cdot)U^\dagger$ fails to satisfy Eq. \eqref{Eq:classical-simulatble} and vice-versa.  Yet, on the level of network nonlocality, our result shows that Clifford operations fail to offer any non-classical advantage. 
\begin{theorem}
\label{Thm:Clifford}
Every Clifford $k$-network correlation is $k$-network local.
\end{theorem}

\begin{figure}[t]
\centering
\includegraphics[scale=0.69]{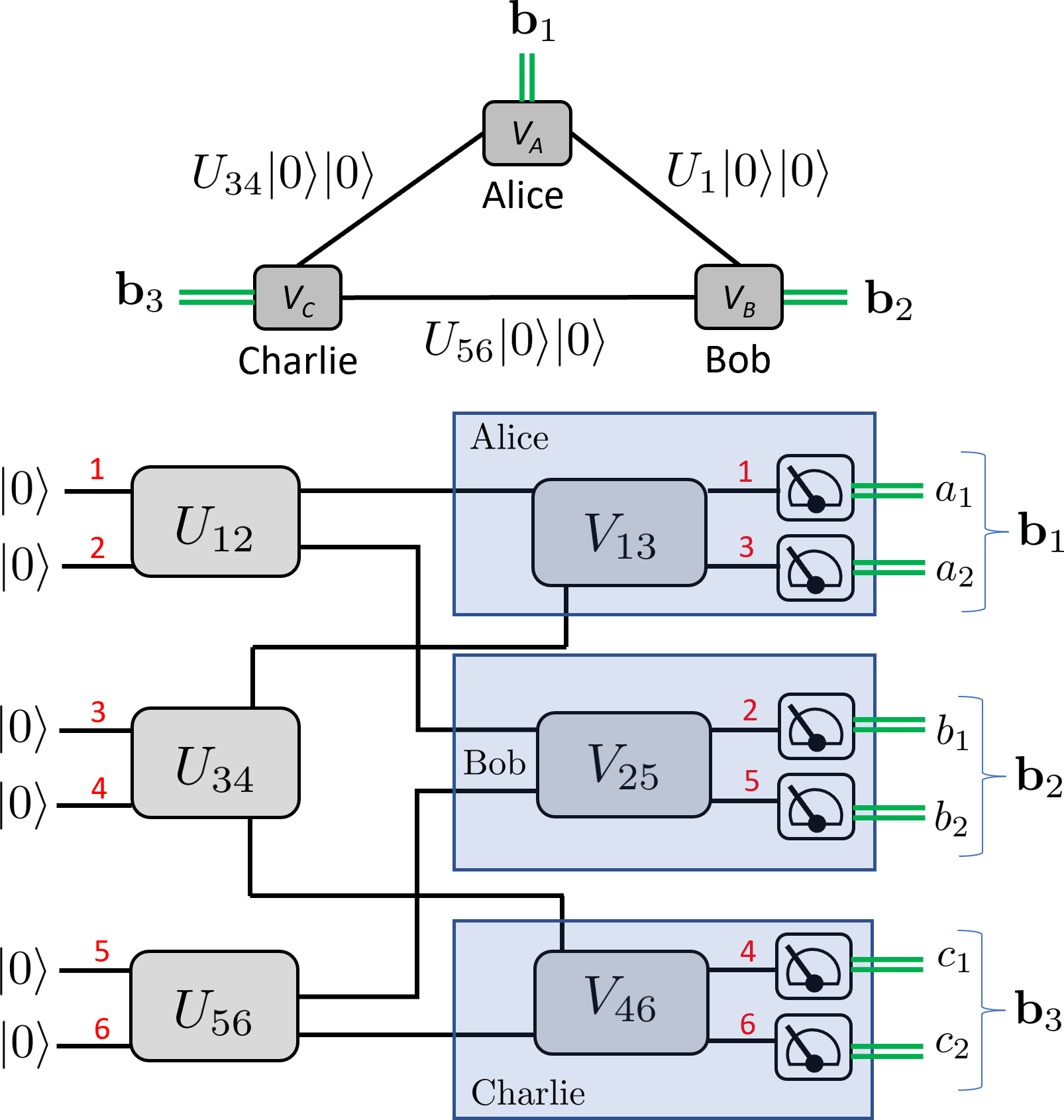}
\caption{The triangle network represented in the standard circuit model, with no ancilla qubits.  The distribution over $(a_1,a_2,b_1,b_2,c_1,c_2)$ is a Clifford $2$-network distribution if each of the $U_{e}$ and $V_{i}$ are Clifford gates.}
\label{Fig:network-pic}
\end{figure}

The proof of this theorem begins by performing three simplifications.  %First, we observe that it suffices to consider the distribution of pure stabilizer states $\rho_e=\op{\varphi_e}{\varphi_e}$ along each hyperedge in a Clifford $k$-network.  Indeed, the $k$-party randomness in each mixed state $\rho_e$ can be absorbed into the classical randomness $\lambda_e$ simulating the distribution (see the Appendix for more details); hence if every Clifford $k$-network correlation is $k$-network local for pure states, then it will also be for mixed states.  
First, in Proposition \ref{Prop:k-to-2} we established that every quantum $k$-network nonlocal distribution can be obtained from a quantum $2$-network nonlocal distribution via post-selection.  Moreover, since teleportation is carried out using Pauli gates, the proof of Proposition \ref{Prop:k-to-2} can be specialized to Clifford networks: if a quantum $k$-network nonlocal distribution is generated on a Clifford $k$-network, then there exists a quantum $2$-network nonlocal distribution generated on a Clifford $2$-network.  Hence, to prove Theorem \ref{Thm:Clifford}, it suffices to show that measuring Clifford $2$-network states always leads to $2$-network local distributions.  In other words, we can restrict $\mc{G}$ to just being a graph so that each $\ket{\varphi_e}$ is a bipartite stabilizer state.  Second, recall the fact that every bipartite stabilizer state can be transformed into copies of $\ket{\Phi^+}=\frac{1}{\sqrt{2}}(\ket{00}+\ket{11})$ and computational basis states $\ket{0}$ using local Cliffords \cite{Fattal-2004a}.  Third, as proved in the Supplemental Material, the use of local ancilla provides no advantage for the purpose of generating $2$-network nonlocality. Then it suffices to consider graphs $\mc{G}$ in which each node is connected to exactly one other node, the two nodes are held by different parties, and their connecting edge represents a maximally entangled state $\ket{\Phi^+}$ shared between them.  %We denote the corresponding graph state as $\ket{\Phi^+_{\mc{G}}}:=\bigotimes_{e\in E}\ket{\Phi^+_e}$.

The next part of the proof involves constructing a local model for any graph having this structure and any choice of local stabilizer measurement.  Our local model involves replacing each edge $e$ on the graph with a random variable $\lambda_e=(\ol{X}_e,\ol{Y}_e,\ol{Z}_e)$, which is a trio of uniform independent random bits.  Based on the local measurement $\msf{A}_i$ wants to perform and the values $\{\lambda_e\}_{e\in A_i}$, outputs $\mbf{b}_i$ can be generated that jointly have the correct distribution $p(\mbf{b}_1,\cdots,\mbf{b}_N)$.  Details of the model and a proof of its correctness are presented in the Supplemental Material.  During the completion of this manuscript we became aware of a similar model independently derived by Matthew Pusey in Ref. \cite{Pusey2010}, although it does not directly account for network constraints and the emergence of network nonlocality \footnote{One feature of Pusey's model in Ref. \cite{Pusey2010} that can also be applied here is a slight simplification of the shared classical randomness.  Instead of needing a trio of random bits $\lambda_e=(\ol{X}_e,\ol{Y}_e,\ol{Z}_e)$ for each edge $e$, it suffices to have two bits $(\ol{X}_e,\ol{Z}_e)$ and then compute the third random bit as their product, $\ol{Y}_e=\ol{X}_e\ol{Z}_e$.}.

\textit{Discussion.}  In this letter, we have presented a unifying framework for the general study of quantum network nonlocality.  We believe this framework can help clarify what types of quantum resources are needed to generate different forms of nonlocality.  It also helps draw a trifold connection between nonlocality, quantum resource theories, and quantum computation, as shown in Fig. \ref{Fig:network-pic}.  When any protocol is fully decomposed as a distributed quantum circuit, we found that some non-Clifford operation is needed to generate nonlocal correlations. 

The importance of non-Clifford operations is further corroborated by the purported examples of genuine network nonlocality or full network nonlocality found in the literature \cite{Renou-2019b, Tavakoli-2021a, Renou-2022a, Renou-2022b, Pozas-Kerstjens-2022a, Pozas-Kerstjens-2022b}. In each of these examples some non-Clifford operation is needed to realize the nonlocality. Even more conspicuously, in the triangle network considered by Renou \textit{et al.} in Ref. \cite{Renou-2019b}, nonlocality fails to emerge exactly when the parameters of their model coincided with a Clifford network. Our work formalizes the reason why this happens, and can thus be interpreted as a guide for what kind of resources are needed to generate network nonlocality.

One may feel that Theorem \ref{Thm:Clifford} is not that surprising given the Gottesman-Knill Theorem for simulating Clifford circuits.  However, let us highlight two reasons why the former stands independently of the latter, despite them sharing a kindred spirit.   First, the Gottesman-Knill Theorem provides a classical algorithm to correctly reproduce the outcome statistics when locally measuring a stabilizer state.  However, in this algorithm, one must update the generators of the stabilizer after simulating the measurement of each qubit.  This requires the communication of global information, which is forbidden in the network nonlocality model.  Therefore, a completely new classical model is needed for a \textit{distributed} simulation, which is what Theorem \ref{Thm:Clifford} provides. %\amanda{\textbf{(FOOTNOTE)}If one allows for some communication, then the set of quantum correlations that can be simulated of course becomes large.  For example, Ref. \cite{Tessier2005} modifies the Gottesman-Knill algorithm by allowing for classical communication between qubits in a quantum circuit.  This allows for the reproduction of certain nonlocal correlations generated by measuring GHZ states. Note that $n$-party GHZ states cannot be generated in a Clifford $k$-network}.

Second, Theorem \ref{Thm:Clifford} breaks down when relaxing certain operational restrictions whereas the Gottesman-Knill Theorem does not.  Specifically, suppose we relax condition (ic) by allowing for \textit{mixed} stabilizer states; i.e. convex combinations of stabilizer states whose purification is no longer a stabilizer state.  Remarkably, now it becomes possible to violate Theorem \ref{Thm:Clifford} using so-called ``disguised'' Bell nonlocality \cite{Tavakoli-2022a}.  The construction is presented in the Supplemental Material and involves extending Mermin's magic square game \cite{Mermin-1990a} to the network setting using a slight modification of Fritz's construction.  This result highlights the strong non-convexity that emerges in network nonlocality: local models can be constructed for all pure stabilizer states, but this fails to be possible when taking their mixture. It is tempting to think that the classical randomness in the mixed state $\rho_e$ could just be absorbed into the classical randomness distributed on edge $e$.  However, this does not always work since the local classical strategy for some party $\msf{A}_i$ in simulating a Clifford network distribution might depend on which state $\ket{\varphi_e}$ is seeded on edge $e$, even if $\msf{A}_i$ is not connected on $e$.  

Similarly, one could consider relaxing condition (iic) by allowing for convex combinations of Clifford gates; i.e. $\mc{E}=\sum_kp_kU_k(\cdot)U_k^\dagger$ with each $U_k$ being a Clifford.  Then by a similar construction to the one used for mixed stabilizer states, it is is possible to generate nonlocality on the network.  The problem is that not every Clifford channel like $\mc{E}$ admits a unitary diltation that itself is Clifford, and thus Theorem \ref{Thm:Clifford} does not apply.  In contrast, the Gottesman-Knill Theorem still provides an efficient classical simulation algorithm for quantum circuits using mixed stabilizer states and random Clifford channels \cite{Veitch-2014a}.  Ultimately, we hope that these findings and the framework described here help pinpoint the precise conditions necessary for generating nonlocal correlations on a quantum network.

\textit{Acknowledgments --} We thank Marc-Olivier Renou for discussing with us the subtleties of disguised network nonlocality.  We are also grateful to Mark Howard, Jonathan Barrett, Nicolas Brunner, and Matthew Pusey for bringing several precursor works to our attention and discussing other connections between Clifford circuits and local classical models.    This material is based upon work supported by the U.S. Department of Energy Office of Science National Quantum Information Science Research Centers.

\bibliography{networks}

%apsrev4-2.bst 2019-01-14 (MD) hand-edited version of apsrev4-1.bst
%Control: key (0)
%Control: author (8) initials jnrlst
%Control: editor formatted (1) identically to author
%Control: production of article title (0) allowed
%Control: page (0) single
%Control: year (1) truncated
%Control: production of eprint (0) enabled
\begin{thebibliography}{40}%
\makeatletter
\providecommand \@ifxundefined [1]{%
 \@ifx{#1\undefined}
}%
\providecommand \@ifnum [1]{%
 \ifnum #1\expandafter \@firstoftwo
 \else \expandafter \@secondoftwo
 \fi
}%
\providecommand \@ifx [1]{%
 \ifx #1\expandafter \@firstoftwo
 \else \expandafter \@secondoftwo
 \fi
}%
\providecommand \natexlab [1]{#1}%
\providecommand \enquote  [1]{``#1''}%
\providecommand \bibnamefont  [1]{#1}%
\providecommand \bibfnamefont [1]{#1}%
\providecommand \citenamefont [1]{#1}%
\providecommand \href@noop [0]{\@secondoftwo}%
\providecommand \href [0]{\begingroup \@sanitize@url \@href}%
\providecommand \@href[1]{\@@startlink{#1}\@@href}%
\providecommand \@@href[1]{\endgroup#1\@@endlink}%
\providecommand \@sanitize@url [0]{\catcode `\\12\catcode `\$12\catcode
  `\&12\catcode `\#12\catcode `\^12\catcode `\_12\catcode `\%12\relax}%
\providecommand \@@startlink[1]{}%
\providecommand \@@endlink[0]{}%
\providecommand \url  [0]{\begingroup\@sanitize@url \@url }%
\providecommand \@url [1]{\endgroup\@href {#1}{\urlprefix }}%
\providecommand \urlprefix  [0]{URL }%
\providecommand \Eprint [0]{\href }%
\providecommand \doibase [0]{https://doi.org/}%
\providecommand \selectlanguage [0]{\@gobble}%
\providecommand \bibinfo  [0]{\@secondoftwo}%
\providecommand \bibfield  [0]{\@secondoftwo}%
\providecommand \translation [1]{[#1]}%
\providecommand \BibitemOpen [0]{}%
\providecommand \bibitemStop [0]{}%
\providecommand \bibitemNoStop [0]{.\EOS\space}%
\providecommand \EOS [0]{\spacefactor3000\relax}%
\providecommand \BibitemShut  [1]{\csname bibitem#1\endcsname}%
\let\auto@bib@innerbib\@empty
%</preamble>
\bibitem [{\citenamefont {Gottesman}(1998)}]{Gottesman-1998a}%
  \BibitemOpen
  \bibfield  {author} {\bibinfo {author} {\bibfnamefont {D.}~\bibnamefont
  {Gottesman}},\ }\href@noop {} {\bibinfo {title} {The heisenberg
  representation of quantum computers}} (\bibinfo {year} {1998}),\ \Eprint
  {https://arxiv.org/abs/arXiv:quant-ph/9807006} {arXiv:quant-ph/9807006}
  \BibitemShut {NoStop}%
\bibitem [{\citenamefont {Aaronson}\ and\ \citenamefont
  {Gottesman}(2004)}]{Aaronson-2004a}%
  \BibitemOpen
  \bibfield  {author} {\bibinfo {author} {\bibfnamefont {S.}~\bibnamefont
  {Aaronson}}\ and\ \bibinfo {author} {\bibfnamefont {D.}~\bibnamefont
  {Gottesman}},\ }\bibfield  {title} {\bibinfo {title} {Improved simulation of
  stabilizer circuits},\ }\href {https://doi.org/10.1103/PhysRevA.70.052328}
  {\bibfield  {journal} {\bibinfo  {journal} {Phys. Rev. A}\ }\textbf {\bibinfo
  {volume} {70}},\ \bibinfo {pages} {052328} (\bibinfo {year}
  {2004})}\BibitemShut {NoStop}%
\bibitem [{\citenamefont {Anders}\ and\ \citenamefont
  {Briegel}(2006)}]{Anders-2006a}%
  \BibitemOpen
  \bibfield  {author} {\bibinfo {author} {\bibfnamefont {S.}~\bibnamefont
  {Anders}}\ and\ \bibinfo {author} {\bibfnamefont {H.~J.}\ \bibnamefont
  {Briegel}},\ }\bibfield  {title} {\bibinfo {title} {Fast simulation of
  stabilizer circuits using a graph-state representation},\ }\href
  {https://doi.org/10.1103/PhysRevA.73.022334} {\bibfield  {journal} {\bibinfo
  {journal} {Phys. Rev. A}\ }\textbf {\bibinfo {volume} {73}},\ \bibinfo
  {pages} {022334} (\bibinfo {year} {2006})}\BibitemShut {NoStop}%
\bibitem [{\citenamefont {Pusey}(2010)}]{Pusey2010}%
  \BibitemOpen
  \bibfield  {author} {\bibinfo {author} {\bibfnamefont {M.}~\bibnamefont
  {Pusey}},\ }\emph {\bibinfo {title} {A few connections between Quantum
  Computation and Quantum Non-Locality}},\ \href@noop {} {Master's thesis},\
  \bibinfo  {school} {Imperial College London} (\bibinfo {year}
  {2010})\BibitemShut {NoStop}%
\bibitem [{\citenamefont {Tessier}\ \emph {et~al.}(2005)\citenamefont
  {Tessier}, \citenamefont {Caves}, \citenamefont {Deutsch}, \citenamefont
  {Eastin},\ and\ \citenamefont {Bacon}}]{Tessier2005}%
  \BibitemOpen
  \bibfield  {author} {\bibinfo {author} {\bibfnamefont {T.~E.}\ \bibnamefont
  {Tessier}}, \bibinfo {author} {\bibfnamefont {C.~M.}\ \bibnamefont {Caves}},
  \bibinfo {author} {\bibfnamefont {I.~H.}\ \bibnamefont {Deutsch}}, \bibinfo
  {author} {\bibfnamefont {B.}~\bibnamefont {Eastin}},\ and\ \bibinfo {author}
  {\bibfnamefont {D.}~\bibnamefont {Bacon}},\ }\bibfield  {title} {\bibinfo
  {title} {Optimal classical-communication-assisted local model of $n$-qubit
  greenberger-horne-zeilinger correlations},\ }\href
  {https://doi.org/10.1103/PhysRevA.72.032305} {\bibfield  {journal} {\bibinfo
  {journal} {Phys. Rev. A}\ }\textbf {\bibinfo {volume} {72}},\ \bibinfo
  {pages} {032305} (\bibinfo {year} {2005})}\BibitemShut {NoStop}%
\bibitem [{\citenamefont {Howard}\ and\ \citenamefont
  {Vala}(2012)}]{Howard2012}%
  \BibitemOpen
  \bibfield  {author} {\bibinfo {author} {\bibfnamefont {M.}~\bibnamefont
  {Howard}}\ and\ \bibinfo {author} {\bibfnamefont {J.}~\bibnamefont {Vala}},\
  }\bibfield  {title} {\bibinfo {title} {Nonlocality as a benchmark for
  universal quantum computation in ising anyon topological quantum computers},\
  }\href {https://doi.org/10.1103/PhysRevA.85.022304} {\bibfield  {journal}
  {\bibinfo  {journal} {Phys. Rev. A}\ }\textbf {\bibinfo {volume} {85}},\
  \bibinfo {pages} {022304} (\bibinfo {year} {2012})}\BibitemShut {NoStop}%
\bibitem [{\citenamefont {Howard}(2015)}]{Howard2015}%
  \BibitemOpen
  \bibfield  {author} {\bibinfo {author} {\bibfnamefont {M.}~\bibnamefont
  {Howard}},\ }\bibfield  {title} {\bibinfo {title} {Maximum nonlocality and
  minimum uncertainty using magic states},\ }\href
  {https://doi.org/10.1103/PhysRevA.91.042103} {\bibfield  {journal} {\bibinfo
  {journal} {Phys. Rev. A}\ }\textbf {\bibinfo {volume} {91}},\ \bibinfo
  {pages} {042103} (\bibinfo {year} {2015})}\BibitemShut {NoStop}%
\bibitem [{\citenamefont {Brunner}\ \emph {et~al.}(2014)\citenamefont
  {Brunner}, \citenamefont {Cavalcanti}, \citenamefont {Pironio}, \citenamefont
  {Scarani},\ and\ \citenamefont {Wehner}}]{Brunner-2014a}%
  \BibitemOpen
  \bibfield  {author} {\bibinfo {author} {\bibfnamefont {N.}~\bibnamefont
  {Brunner}}, \bibinfo {author} {\bibfnamefont {D.}~\bibnamefont {Cavalcanti}},
  \bibinfo {author} {\bibfnamefont {S.}~\bibnamefont {Pironio}}, \bibinfo
  {author} {\bibfnamefont {V.}~\bibnamefont {Scarani}},\ and\ \bibinfo {author}
  {\bibfnamefont {S.}~\bibnamefont {Wehner}},\ }\bibfield  {title} {\bibinfo
  {title} {Bell nonlocality},\ }\href
  {https://doi.org/10.1103/RevModPhys.86.419} {\bibfield  {journal} {\bibinfo
  {journal} {Rev. Mod. Phys.}\ }\textbf {\bibinfo {volume} {86}},\ \bibinfo
  {pages} {419} (\bibinfo {year} {2014})}\BibitemShut {NoStop}%
\bibitem [{\citenamefont {Branciard}\ \emph {et~al.}(2010)\citenamefont
  {Branciard}, \citenamefont {Gisin},\ and\ \citenamefont
  {Pironio}}]{Branciard-2010a}%
  \BibitemOpen
  \bibfield  {author} {\bibinfo {author} {\bibfnamefont {C.}~\bibnamefont
  {Branciard}}, \bibinfo {author} {\bibfnamefont {N.}~\bibnamefont {Gisin}},\
  and\ \bibinfo {author} {\bibfnamefont {S.}~\bibnamefont {Pironio}},\
  }\bibfield  {title} {\bibinfo {title} {Characterizing the nonlocal
  correlations created via entanglement swapping},\ }\href
  {https://doi.org/10.1103/PhysRevLett.104.170401} {\bibfield  {journal}
  {\bibinfo  {journal} {Phys. Rev. Lett.}\ }\textbf {\bibinfo {volume} {104}},\
  \bibinfo {pages} {170401} (\bibinfo {year} {2010})}\BibitemShut {NoStop}%
\bibitem [{\citenamefont {Bancal}\ \emph {et~al.}(2011)\citenamefont {Bancal},
  \citenamefont {Brunner}, \citenamefont {Gisin},\ and\ \citenamefont
  {Liang}}]{Bancal-2011a}%
  \BibitemOpen
  \bibfield  {author} {\bibinfo {author} {\bibfnamefont {J.-D.}\ \bibnamefont
  {Bancal}}, \bibinfo {author} {\bibfnamefont {N.}~\bibnamefont {Brunner}},
  \bibinfo {author} {\bibfnamefont {N.}~\bibnamefont {Gisin}},\ and\ \bibinfo
  {author} {\bibfnamefont {Y.-C.}\ \bibnamefont {Liang}},\ }\bibfield  {title}
  {\bibinfo {title} {Detecting genuine multipartite quantum nonlocality: A
  simple approach and generalization to arbitrary dimensions},\ }\href
  {https://doi.org/10.1103/PhysRevLett.106.020405} {\bibfield  {journal}
  {\bibinfo  {journal} {Phys. Rev. Lett.}\ }\textbf {\bibinfo {volume} {106}},\
  \bibinfo {pages} {020405} (\bibinfo {year} {2011})}\BibitemShut {NoStop}%
\bibitem [{\citenamefont {Branciard}\ \emph {et~al.}(2012)\citenamefont
  {Branciard}, \citenamefont {Rosset}, \citenamefont {Gisin},\ and\
  \citenamefont {Pironio}}]{Branciard-2012a}%
  \BibitemOpen
  \bibfield  {author} {\bibinfo {author} {\bibfnamefont {C.}~\bibnamefont
  {Branciard}}, \bibinfo {author} {\bibfnamefont {D.}~\bibnamefont {Rosset}},
  \bibinfo {author} {\bibfnamefont {N.}~\bibnamefont {Gisin}},\ and\ \bibinfo
  {author} {\bibfnamefont {S.}~\bibnamefont {Pironio}},\ }\bibfield  {title}
  {\bibinfo {title} {Bilocal versus nonbilocal correlations in
  entanglement-swapping experiments},\ }\href
  {https://doi.org/10.1103/PhysRevA.85.032119} {\bibfield  {journal} {\bibinfo
  {journal} {Phys. Rev. A}\ }\textbf {\bibinfo {volume} {85}},\ \bibinfo
  {pages} {032119} (\bibinfo {year} {2012})}\BibitemShut {NoStop}%
\bibitem [{\citenamefont {Gisin}(2019)}]{Gisin-2019a}%
  \BibitemOpen
  \bibfield  {author} {\bibinfo {author} {\bibfnamefont {N.}~\bibnamefont
  {Gisin}},\ }\bibfield  {title} {\bibinfo {title} {Entanglement 25 years after
  quantum teleportation: Testing joint measurements in quantum networks},\
  }\href {https://doi.org/10.3390/e21030325} {\bibfield  {journal} {\bibinfo
  {journal} {Entropy}\ }\textbf {\bibinfo {volume} {21}},\ \bibinfo {pages}
  {325} (\bibinfo {year} {2019})}\BibitemShut {NoStop}%
\bibitem [{\citenamefont {Tavakoli}\ \emph {et~al.}(2022)\citenamefont
  {Tavakoli}, \citenamefont {Pozas-Kerstjens}, \citenamefont {Luo},\ and\
  \citenamefont {Renou}}]{Tavakoli-2022a}%
  \BibitemOpen
  \bibfield  {author} {\bibinfo {author} {\bibfnamefont {A.}~\bibnamefont
  {Tavakoli}}, \bibinfo {author} {\bibfnamefont {A.}~\bibnamefont
  {Pozas-Kerstjens}}, \bibinfo {author} {\bibfnamefont {M.-X.}\ \bibnamefont
  {Luo}},\ and\ \bibinfo {author} {\bibfnamefont {M.-O.}\ \bibnamefont
  {Renou}},\ }\bibfield  {title} {\bibinfo {title} {Bell nonlocality in
  networks},\ }\href {https://doi.org/10.1088/1361-6633/ac41bb} {\bibfield
  {journal} {\bibinfo  {journal} {Reports on Progress in Physics}\ }\textbf
  {\bibinfo {volume} {85}},\ \bibinfo {pages} {056001} (\bibinfo {year}
  {2022})}\BibitemShut {NoStop}%
\bibitem [{\citenamefont {Renou}\ \emph {et~al.}(2019)\citenamefont {Renou},
  \citenamefont {B\"aumer}, \citenamefont {Boreiri}, \citenamefont {Brunner},
  \citenamefont {Gisin},\ and\ \citenamefont {Beigi}}]{Renou-2019b}%
  \BibitemOpen
  \bibfield  {author} {\bibinfo {author} {\bibfnamefont {M.-O.}\ \bibnamefont
  {Renou}}, \bibinfo {author} {\bibfnamefont {E.}~\bibnamefont {B\"aumer}},
  \bibinfo {author} {\bibfnamefont {S.}~\bibnamefont {Boreiri}}, \bibinfo
  {author} {\bibfnamefont {N.}~\bibnamefont {Brunner}}, \bibinfo {author}
  {\bibfnamefont {N.}~\bibnamefont {Gisin}},\ and\ \bibinfo {author}
  {\bibfnamefont {S.}~\bibnamefont {Beigi}},\ }\bibfield  {title} {\bibinfo
  {title} {Genuine quantum nonlocality in the triangle network},\ }\href
  {https://doi.org/10.1103/PhysRevLett.123.140401} {\bibfield  {journal}
  {\bibinfo  {journal} {Phys. Rev. Lett.}\ }\textbf {\bibinfo {volume} {123}},\
  \bibinfo {pages} {140401} (\bibinfo {year} {2019})}\BibitemShut {NoStop}%
\bibitem [{\citenamefont {Tavakoli}\ \emph {et~al.}(2021)\citenamefont
  {Tavakoli}, \citenamefont {Gisin},\ and\ \citenamefont
  {Branciard}}]{Tavakoli-2021a}%
  \BibitemOpen
  \bibfield  {author} {\bibinfo {author} {\bibfnamefont {A.}~\bibnamefont
  {Tavakoli}}, \bibinfo {author} {\bibfnamefont {N.}~\bibnamefont {Gisin}},\
  and\ \bibinfo {author} {\bibfnamefont {C.}~\bibnamefont {Branciard}},\
  }\bibfield  {title} {\bibinfo {title} {Bilocal bell inequalities violated by
  the quantum elegant joint measurement},\ }\href
  {https://doi.org/10.1103/PhysRevLett.126.220401} {\bibfield  {journal}
  {\bibinfo  {journal} {Phys. Rev. Lett.}\ }\textbf {\bibinfo {volume} {126}},\
  \bibinfo {pages} {220401} (\bibinfo {year} {2021})}\BibitemShut {NoStop}%
\bibitem [{\citenamefont {Renou}\ and\ \citenamefont
  {Beigi}(2022{\natexlab{a}})}]{Renou-2022a}%
  \BibitemOpen
  \bibfield  {author} {\bibinfo {author} {\bibfnamefont {M.-O.}\ \bibnamefont
  {Renou}}\ and\ \bibinfo {author} {\bibfnamefont {S.}~\bibnamefont {Beigi}},\
  }\bibfield  {title} {\bibinfo {title} {Network nonlocality via rigidity of
  token counting and color matching},\ }\href
  {https://doi.org/10.1103/PhysRevA.105.022408} {\bibfield  {journal} {\bibinfo
   {journal} {Phys. Rev. A}\ }\textbf {\bibinfo {volume} {105}},\ \bibinfo
  {pages} {022408} (\bibinfo {year} {2022}{\natexlab{a}})}\BibitemShut
  {NoStop}%
\bibitem [{\citenamefont {Renou}\ and\ \citenamefont
  {Beigi}(2022{\natexlab{b}})}]{Renou-2022b}%
  \BibitemOpen
  \bibfield  {author} {\bibinfo {author} {\bibfnamefont {M.-O.}\ \bibnamefont
  {Renou}}\ and\ \bibinfo {author} {\bibfnamefont {S.}~\bibnamefont {Beigi}},\
  }\bibfield  {title} {\bibinfo {title} {Nonlocality for generic networks},\
  }\href {https://doi.org/10.1103/PhysRevLett.128.060401} {\bibfield  {journal}
  {\bibinfo  {journal} {Phys. Rev. Lett.}\ }\textbf {\bibinfo {volume} {128}},\
  \bibinfo {pages} {060401} (\bibinfo {year} {2022}{\natexlab{b}})}\BibitemShut
  {NoStop}%
\bibitem [{\citenamefont {Pozas-Kerstjens}\ \emph
  {et~al.}(2022{\natexlab{a}})\citenamefont {Pozas-Kerstjens}, \citenamefont
  {Gisin},\ and\ \citenamefont {Tavakoli}}]{Pozas-Kerstjens-2022a}%
  \BibitemOpen
  \bibfield  {author} {\bibinfo {author} {\bibfnamefont {A.}~\bibnamefont
  {Pozas-Kerstjens}}, \bibinfo {author} {\bibfnamefont {N.}~\bibnamefont
  {Gisin}},\ and\ \bibinfo {author} {\bibfnamefont {A.}~\bibnamefont
  {Tavakoli}},\ }\bibfield  {title} {\bibinfo {title} {Full network
  nonlocality},\ }\href {https://doi.org/10.1103/PhysRevLett.128.010403}
  {\bibfield  {journal} {\bibinfo  {journal} {Phys. Rev. Lett.}\ }\textbf
  {\bibinfo {volume} {128}},\ \bibinfo {pages} {010403} (\bibinfo {year}
  {2022}{\natexlab{a}})}\BibitemShut {NoStop}%
\bibitem [{\citenamefont {Pozas-Kerstjens}\ \emph
  {et~al.}(2022{\natexlab{b}})\citenamefont {Pozas-Kerstjens}, \citenamefont
  {Gisin},\ and\ \citenamefont {Renou}}]{Pozas-Kerstjens-2022b}%
  \BibitemOpen
  \bibfield  {author} {\bibinfo {author} {\bibfnamefont {A.}~\bibnamefont
  {Pozas-Kerstjens}}, \bibinfo {author} {\bibfnamefont {N.}~\bibnamefont
  {Gisin}},\ and\ \bibinfo {author} {\bibfnamefont {M.-O.}\ \bibnamefont
  {Renou}},\ }\href@noop {} {\bibinfo {title} {Proofs of network quantum
  nonlocality aided by machine learning}} (\bibinfo {year}
  {2022}{\natexlab{b}}),\ \Eprint {https://arxiv.org/abs/arXiv:2203.16543}
  {arXiv:2203.16543} \BibitemShut {NoStop}%
\bibitem [{\citenamefont {Horodecki}\ and\ \citenamefont
  {Oppenheim}(2012)}]{Horodecki-2012a}%
  \BibitemOpen
  \bibfield  {author} {\bibinfo {author} {\bibfnamefont {M.}~\bibnamefont
  {Horodecki}}\ and\ \bibinfo {author} {\bibfnamefont {J.}~\bibnamefont
  {Oppenheim}},\ }\bibfield  {title} {\bibinfo {title} {(quantumness in the
  context of) resource theories},\ }\href
  {https://doi.org/10.1142/s0217979213450197} {\bibfield  {journal} {\bibinfo
  {journal} {International Journal of Modern Physics B}\ }\textbf {\bibinfo
  {volume} {27}},\ \bibinfo {pages} {1345019} (\bibinfo {year}
  {2012})}\BibitemShut {NoStop}%
\bibitem [{\citenamefont {Chitambar}\ and\ \citenamefont
  {Gour}(2019)}]{Chitambar-2019a}%
  \BibitemOpen
  \bibfield  {author} {\bibinfo {author} {\bibfnamefont {E.}~\bibnamefont
  {Chitambar}}\ and\ \bibinfo {author} {\bibfnamefont {G.}~\bibnamefont
  {Gour}},\ }\bibfield  {title} {\bibinfo {title} {Quantum resource theories},\
  }\href {https://doi.org/10.1103/RevModPhys.91.025001} {\bibfield  {journal}
  {\bibinfo  {journal} {Rev. Mod. Phys.}\ }\textbf {\bibinfo {volume} {91}},\
  \bibinfo {pages} {025001} (\bibinfo {year} {2019})}\BibitemShut {NoStop}%
\bibitem [{\citenamefont {Fritz}(2012)}]{Fritz-2012a}%
  \BibitemOpen
  \bibfield  {author} {\bibinfo {author} {\bibfnamefont {T.}~\bibnamefont
  {Fritz}},\ }\bibfield  {title} {\bibinfo {title} {Beyond bell's theorem:
  correlation scenarios},\ }\href
  {https://doi.org/10.1088/1367-2630/14/10/103001} {\bibfield  {journal}
  {\bibinfo  {journal} {New Journal of Physics}\ }\textbf {\bibinfo {volume}
  {14}},\ \bibinfo {pages} {103001} (\bibinfo {year} {2012})}\BibitemShut
  {NoStop}%
\bibitem [{\citenamefont {Liu}\ \emph {et~al.}(2017)\citenamefont {Liu},
  \citenamefont {Hu},\ and\ \citenamefont {Lloyd}}]{Liu-2017a}%
  \BibitemOpen
  \bibfield  {author} {\bibinfo {author} {\bibfnamefont {Z.-W.}\ \bibnamefont
  {Liu}}, \bibinfo {author} {\bibfnamefont {X.}~\bibnamefont {Hu}},\ and\
  \bibinfo {author} {\bibfnamefont {S.}~\bibnamefont {Lloyd}},\ }\bibfield
  {title} {\bibinfo {title} {Resource destroying maps},\ }\href
  {https://doi.org/10.1103/PhysRevLett.118.060502} {\bibfield  {journal}
  {\bibinfo  {journal} {Phys. Rev. Lett.}\ }\textbf {\bibinfo {volume} {118}},\
  \bibinfo {pages} {060502} (\bibinfo {year} {2017})}\BibitemShut {NoStop}%
\bibitem [{\citenamefont {\ifmmode \check{S}\else
  \v{S}\fi{}upi\ifmmode~\acute{c}\else \'{c}\fi{}}\ \emph
  {et~al.}(2022)\citenamefont {\ifmmode \check{S}\else
  \v{S}\fi{}upi\ifmmode~\acute{c}\else \'{c}\fi{}}, \citenamefont {Bancal},
  \citenamefont {Cai},\ and\ \citenamefont {Brunner}}]{Brunner_2022}%
  \BibitemOpen
  \bibfield  {author} {\bibinfo {author} {\bibfnamefont {I.}~\bibnamefont
  {\ifmmode \check{S}\else \v{S}\fi{}upi\ifmmode~\acute{c}\else \'{c}\fi{}}},
  \bibinfo {author} {\bibfnamefont {J.-D.}\ \bibnamefont {Bancal}}, \bibinfo
  {author} {\bibfnamefont {Y.}~\bibnamefont {Cai}},\ and\ \bibinfo {author}
  {\bibfnamefont {N.}~\bibnamefont {Brunner}},\ }\bibfield  {title} {\bibinfo
  {title} {Genuine network quantum nonlocality and self-testing},\ }\href
  {https://doi.org/10.1103/PhysRevA.105.022206} {\bibfield  {journal} {\bibinfo
   {journal} {Phys. Rev. A}\ }\textbf {\bibinfo {volume} {105}},\ \bibinfo
  {pages} {022206} (\bibinfo {year} {2022})}\BibitemShut {NoStop}%
\bibitem [{\citenamefont {Gallego}\ \emph {et~al.}(2012)\citenamefont
  {Gallego}, \citenamefont {W\"urflinger}, \citenamefont {Ac\'{\i}n},\ and\
  \citenamefont {Navascu\'es}}]{Gallego-2012a}%
  \BibitemOpen
  \bibfield  {author} {\bibinfo {author} {\bibfnamefont {R.}~\bibnamefont
  {Gallego}}, \bibinfo {author} {\bibfnamefont {L.~E.}\ \bibnamefont
  {W\"urflinger}}, \bibinfo {author} {\bibfnamefont {A.}~\bibnamefont
  {Ac\'{\i}n}},\ and\ \bibinfo {author} {\bibfnamefont {M.}~\bibnamefont
  {Navascu\'es}},\ }\bibfield  {title} {\bibinfo {title} {Operational framework
  for nonlocality},\ }\href {https://doi.org/10.1103/PhysRevLett.109.070401}
  {\bibfield  {journal} {\bibinfo  {journal} {Phys. Rev. Lett.}\ }\textbf
  {\bibinfo {volume} {109}},\ \bibinfo {pages} {070401} (\bibinfo {year}
  {2012})}\BibitemShut {NoStop}%
\bibitem [{\citenamefont {de~Vicente}(2014)}]{deVicente-2014a}%
  \BibitemOpen
  \bibfield  {author} {\bibinfo {author} {\bibfnamefont {J.~I.}\ \bibnamefont
  {de~Vicente}},\ }\bibfield  {title} {\bibinfo {title} {On nonlocality as a
  resource theory and nonlocality measures},\ }\href
  {https://doi.org/10.1088/1751-8113/47/42/424017} {\bibfield  {journal}
  {\bibinfo  {journal} {Journal of Physics A: Mathematical and Theoretical}\
  }\textbf {\bibinfo {volume} {47}},\ \bibinfo {pages} {424017} (\bibinfo
  {year} {2014})}\BibitemShut {NoStop}%
\bibitem [{\citenamefont {Hensen}\ \emph {et~al.}(2015)\citenamefont {Hensen},
  \citenamefont {Bernien}, \citenamefont {Dr{\'{e}}au}, \citenamefont
  {Reiserer}, \citenamefont {Kalb}, \citenamefont {Blok}, \citenamefont
  {Ruitenberg}, \citenamefont {Vermeulen}, \citenamefont {Schouten},
  \citenamefont {Abell{\'{a}}n}, \citenamefont {Amaya}, \citenamefont
  {Pruneri}, \citenamefont {Mitchell}, \citenamefont {Markham}, \citenamefont
  {Twitchen}, \citenamefont {Elkouss}, \citenamefont {Wehner}, \citenamefont
  {Taminiau},\ and\ \citenamefont {Hanson}}]{Hensen-2015a}%
  \BibitemOpen
  \bibfield  {author} {\bibinfo {author} {\bibfnamefont {B.}~\bibnamefont
  {Hensen}}, \bibinfo {author} {\bibfnamefont {H.}~\bibnamefont {Bernien}},
  \bibinfo {author} {\bibfnamefont {A.~E.}\ \bibnamefont {Dr{\'{e}}au}},
  \bibinfo {author} {\bibfnamefont {A.}~\bibnamefont {Reiserer}}, \bibinfo
  {author} {\bibfnamefont {N.}~\bibnamefont {Kalb}}, \bibinfo {author}
  {\bibfnamefont {M.~S.}\ \bibnamefont {Blok}}, \bibinfo {author}
  {\bibfnamefont {J.}~\bibnamefont {Ruitenberg}}, \bibinfo {author}
  {\bibfnamefont {R.~F.~L.}\ \bibnamefont {Vermeulen}}, \bibinfo {author}
  {\bibfnamefont {R.~N.}\ \bibnamefont {Schouten}}, \bibinfo {author}
  {\bibfnamefont {C.}~\bibnamefont {Abell{\'{a}}n}}, \bibinfo {author}
  {\bibfnamefont {W.}~\bibnamefont {Amaya}}, \bibinfo {author} {\bibfnamefont
  {V.}~\bibnamefont {Pruneri}}, \bibinfo {author} {\bibfnamefont {M.~W.}\
  \bibnamefont {Mitchell}}, \bibinfo {author} {\bibfnamefont {M.}~\bibnamefont
  {Markham}}, \bibinfo {author} {\bibfnamefont {D.~J.}\ \bibnamefont
  {Twitchen}}, \bibinfo {author} {\bibfnamefont {D.}~\bibnamefont {Elkouss}},
  \bibinfo {author} {\bibfnamefont {S.}~\bibnamefont {Wehner}}, \bibinfo
  {author} {\bibfnamefont {T.~H.}\ \bibnamefont {Taminiau}},\ and\ \bibinfo
  {author} {\bibfnamefont {R.}~\bibnamefont {Hanson}},\ }\bibfield  {title}
  {\bibinfo {title} {Loophole-free bell inequality violation using electron
  spins separated by 1.3 kilometres},\ }\href
  {https://doi.org/10.1038/nature15759} {\bibfield  {journal} {\bibinfo
  {journal} {Nature}\ }\textbf {\bibinfo {volume} {526}},\ \bibinfo {pages}
  {682} (\bibinfo {year} {2015})}\BibitemShut {NoStop}%
\bibitem [{\citenamefont {Shalm}\ \emph {et~al.}(2015)\citenamefont {Shalm},
  \citenamefont {Meyer-Scott}, \citenamefont {Christensen}, \citenamefont
  {Bierhorst}, \citenamefont {Wayne}, \citenamefont {Stevens}, \citenamefont
  {Gerrits}, \citenamefont {Glancy}, \citenamefont {Hamel}, \citenamefont
  {Allman}, \citenamefont {Coakley}, \citenamefont {Dyer}, \citenamefont
  {Hodge}, \citenamefont {Lita}, \citenamefont {Verma}, \citenamefont
  {Lambrocco}, \citenamefont {Tortorici}, \citenamefont {Migdall},
  \citenamefont {Zhang}, \citenamefont {Kumor}, \citenamefont {Farr},
  \citenamefont {Marsili}, \citenamefont {Shaw}, \citenamefont {Stern},
  \citenamefont {Abell\'an}, \citenamefont {Amaya}, \citenamefont {Pruneri},
  \citenamefont {Jennewein}, \citenamefont {Mitchell}, \citenamefont {Kwiat},
  \citenamefont {Bienfang}, \citenamefont {Mirin}, \citenamefont {Knill},\ and\
  \citenamefont {Nam}}]{Shalm-2015a}%
  \BibitemOpen
  \bibfield  {author} {\bibinfo {author} {\bibfnamefont {L.~K.}\ \bibnamefont
  {Shalm}}, \bibinfo {author} {\bibfnamefont {E.}~\bibnamefont {Meyer-Scott}},
  \bibinfo {author} {\bibfnamefont {B.~G.}\ \bibnamefont {Christensen}},
  \bibinfo {author} {\bibfnamefont {P.}~\bibnamefont {Bierhorst}}, \bibinfo
  {author} {\bibfnamefont {M.~A.}\ \bibnamefont {Wayne}}, \bibinfo {author}
  {\bibfnamefont {M.~J.}\ \bibnamefont {Stevens}}, \bibinfo {author}
  {\bibfnamefont {T.}~\bibnamefont {Gerrits}}, \bibinfo {author} {\bibfnamefont
  {S.}~\bibnamefont {Glancy}}, \bibinfo {author} {\bibfnamefont {D.~R.}\
  \bibnamefont {Hamel}}, \bibinfo {author} {\bibfnamefont {M.~S.}\ \bibnamefont
  {Allman}}, \bibinfo {author} {\bibfnamefont {K.~J.}\ \bibnamefont {Coakley}},
  \bibinfo {author} {\bibfnamefont {S.~D.}\ \bibnamefont {Dyer}}, \bibinfo
  {author} {\bibfnamefont {C.}~\bibnamefont {Hodge}}, \bibinfo {author}
  {\bibfnamefont {A.~E.}\ \bibnamefont {Lita}}, \bibinfo {author}
  {\bibfnamefont {V.~B.}\ \bibnamefont {Verma}}, \bibinfo {author}
  {\bibfnamefont {C.}~\bibnamefont {Lambrocco}}, \bibinfo {author}
  {\bibfnamefont {E.}~\bibnamefont {Tortorici}}, \bibinfo {author}
  {\bibfnamefont {A.~L.}\ \bibnamefont {Migdall}}, \bibinfo {author}
  {\bibfnamefont {Y.}~\bibnamefont {Zhang}}, \bibinfo {author} {\bibfnamefont
  {D.~R.}\ \bibnamefont {Kumor}}, \bibinfo {author} {\bibfnamefont {W.~H.}\
  \bibnamefont {Farr}}, \bibinfo {author} {\bibfnamefont {F.}~\bibnamefont
  {Marsili}}, \bibinfo {author} {\bibfnamefont {M.~D.}\ \bibnamefont {Shaw}},
  \bibinfo {author} {\bibfnamefont {J.~A.}\ \bibnamefont {Stern}}, \bibinfo
  {author} {\bibfnamefont {C.}~\bibnamefont {Abell\'an}}, \bibinfo {author}
  {\bibfnamefont {W.}~\bibnamefont {Amaya}}, \bibinfo {author} {\bibfnamefont
  {V.}~\bibnamefont {Pruneri}}, \bibinfo {author} {\bibfnamefont
  {T.}~\bibnamefont {Jennewein}}, \bibinfo {author} {\bibfnamefont {M.~W.}\
  \bibnamefont {Mitchell}}, \bibinfo {author} {\bibfnamefont {P.~G.}\
  \bibnamefont {Kwiat}}, \bibinfo {author} {\bibfnamefont {J.~C.}\ \bibnamefont
  {Bienfang}}, \bibinfo {author} {\bibfnamefont {R.~P.}\ \bibnamefont {Mirin}},
  \bibinfo {author} {\bibfnamefont {E.}~\bibnamefont {Knill}},\ and\ \bibinfo
  {author} {\bibfnamefont {S.~W.}\ \bibnamefont {Nam}},\ }\bibfield  {title}
  {\bibinfo {title} {Strong loophole-free test of local realism},\ }\href
  {https://doi.org/10.1103/PhysRevLett.115.250402} {\bibfield  {journal}
  {\bibinfo  {journal} {Phys. Rev. Lett.}\ }\textbf {\bibinfo {volume} {115}},\
  \bibinfo {pages} {250402} (\bibinfo {year} {2015})}\BibitemShut {NoStop}%
\bibitem [{\citenamefont {Giustina}\ \emph {et~al.}(2015)\citenamefont
  {Giustina}, \citenamefont {Versteegh}, \citenamefont {Wengerowsky},
  \citenamefont {Handsteiner}, \citenamefont {Hochrainer}, \citenamefont
  {Phelan}, \citenamefont {Steinlechner}, \citenamefont {Kofler}, \citenamefont
  {Larsson}, \citenamefont {Abell\'an}, \citenamefont {Amaya}, \citenamefont
  {Pruneri}, \citenamefont {Mitchell}, \citenamefont {Beyer}, \citenamefont
  {Gerrits}, \citenamefont {Lita}, \citenamefont {Shalm}, \citenamefont {Nam},
  \citenamefont {Scheidl}, \citenamefont {Ursin}, \citenamefont {Wittmann},\
  and\ \citenamefont {Zeilinger}}]{Giustina-2015a}%
  \BibitemOpen
  \bibfield  {author} {\bibinfo {author} {\bibfnamefont {M.}~\bibnamefont
  {Giustina}}, \bibinfo {author} {\bibfnamefont {M.~A.~M.}\ \bibnamefont
  {Versteegh}}, \bibinfo {author} {\bibfnamefont {S.}~\bibnamefont
  {Wengerowsky}}, \bibinfo {author} {\bibfnamefont {J.}~\bibnamefont
  {Handsteiner}}, \bibinfo {author} {\bibfnamefont {A.}~\bibnamefont
  {Hochrainer}}, \bibinfo {author} {\bibfnamefont {K.}~\bibnamefont {Phelan}},
  \bibinfo {author} {\bibfnamefont {F.}~\bibnamefont {Steinlechner}}, \bibinfo
  {author} {\bibfnamefont {J.}~\bibnamefont {Kofler}}, \bibinfo {author}
  {\bibfnamefont {J.-A.}\ \bibnamefont {Larsson}}, \bibinfo {author}
  {\bibfnamefont {C.}~\bibnamefont {Abell\'an}}, \bibinfo {author}
  {\bibfnamefont {W.}~\bibnamefont {Amaya}}, \bibinfo {author} {\bibfnamefont
  {V.}~\bibnamefont {Pruneri}}, \bibinfo {author} {\bibfnamefont {M.~W.}\
  \bibnamefont {Mitchell}}, \bibinfo {author} {\bibfnamefont {J.}~\bibnamefont
  {Beyer}}, \bibinfo {author} {\bibfnamefont {T.}~\bibnamefont {Gerrits}},
  \bibinfo {author} {\bibfnamefont {A.~E.}\ \bibnamefont {Lita}}, \bibinfo
  {author} {\bibfnamefont {L.~K.}\ \bibnamefont {Shalm}}, \bibinfo {author}
  {\bibfnamefont {S.~W.}\ \bibnamefont {Nam}}, \bibinfo {author} {\bibfnamefont
  {T.}~\bibnamefont {Scheidl}}, \bibinfo {author} {\bibfnamefont
  {R.}~\bibnamefont {Ursin}}, \bibinfo {author} {\bibfnamefont
  {B.}~\bibnamefont {Wittmann}},\ and\ \bibinfo {author} {\bibfnamefont
  {A.}~\bibnamefont {Zeilinger}},\ }\bibfield  {title} {\bibinfo {title}
  {Significant-loophole-free test of bell's theorem with entangled photons},\
  }\href {https://doi.org/10.1103/PhysRevLett.115.250401} {\bibfield  {journal}
  {\bibinfo  {journal} {Phys. Rev. Lett.}\ }\textbf {\bibinfo {volume} {115}},\
  \bibinfo {pages} {250401} (\bibinfo {year} {2015})}\BibitemShut {NoStop}%
\bibitem [{\citenamefont {Horodecki}\ \emph {et~al.}(2009)\citenamefont
  {Horodecki}, \citenamefont {Horodecki}, \citenamefont {Horodecki},\ and\
  \citenamefont {Horodecki}}]{Horodecki-2009a}%
  \BibitemOpen
  \bibfield  {author} {\bibinfo {author} {\bibfnamefont {R.}~\bibnamefont
  {Horodecki}}, \bibinfo {author} {\bibfnamefont {P.}~\bibnamefont
  {Horodecki}}, \bibinfo {author} {\bibfnamefont {M.}~\bibnamefont
  {Horodecki}},\ and\ \bibinfo {author} {\bibfnamefont {K.}~\bibnamefont
  {Horodecki}},\ }\bibfield  {title} {\bibinfo {title} {Quantum entanglement},\
  }\href {https://doi.org/10.1103/RevModPhys.81.865} {\bibfield  {journal}
  {\bibinfo  {journal} {Rev. Mod. Phys.}\ }\textbf {\bibinfo {volume} {81}},\
  \bibinfo {pages} {865} (\bibinfo {year} {2009})}\BibitemShut {NoStop}%
\bibitem [{\citenamefont {Gisin}\ \emph {et~al.}(2020)\citenamefont {Gisin},
  \citenamefont {Bancal}, \citenamefont {Cai}, \citenamefont {Remy},
  \citenamefont {Tavakoli}, \citenamefont {Cruzeiro}, \citenamefont {Popescu},\
  and\ \citenamefont {Brunner}}]{Gisin-2020a}%
  \BibitemOpen
  \bibfield  {author} {\bibinfo {author} {\bibfnamefont {N.}~\bibnamefont
  {Gisin}}, \bibinfo {author} {\bibfnamefont {J.-D.}\ \bibnamefont {Bancal}},
  \bibinfo {author} {\bibfnamefont {Y.}~\bibnamefont {Cai}}, \bibinfo {author}
  {\bibfnamefont {P.}~\bibnamefont {Remy}}, \bibinfo {author} {\bibfnamefont
  {A.}~\bibnamefont {Tavakoli}}, \bibinfo {author} {\bibfnamefont {E.~Z.}\
  \bibnamefont {Cruzeiro}}, \bibinfo {author} {\bibfnamefont {S.}~\bibnamefont
  {Popescu}},\ and\ \bibinfo {author} {\bibfnamefont {N.}~\bibnamefont
  {Brunner}},\ }\bibfield  {title} {\bibinfo {title} {Constraints on
  nonlocality in networks from no-signaling and independence},\ }\bibfield
  {journal} {\bibinfo  {journal} {Nature Communications}\ }\textbf {\bibinfo
  {volume} {11}},\ \href {https://doi.org/10.1038/s41467-020-16137-4}
  {10.1038/s41467-020-16137-4} (\bibinfo {year} {2020})\BibitemShut {NoStop}%
\bibitem [{\citenamefont {Bierhorst}(2021)}]{Bierhorst-2021a}%
  \BibitemOpen
  \bibfield  {author} {\bibinfo {author} {\bibfnamefont {P.}~\bibnamefont
  {Bierhorst}},\ }\bibfield  {title} {\bibinfo {title} {Ruling out bipartite
  nonsignaling nonlocal models for tripartite correlations},\ }\href
  {https://doi.org/10.1103/PhysRevA.104.012210} {\bibfield  {journal} {\bibinfo
   {journal} {Phys. Rev. A}\ }\textbf {\bibinfo {volume} {104}},\ \bibinfo
  {pages} {012210} (\bibinfo {year} {2021})}\BibitemShut {NoStop}%
\bibitem [{\citenamefont {Gottesman}(1997)}]{Gottesman-1997a}%
  \BibitemOpen
  \bibfield  {author} {\bibinfo {author} {\bibfnamefont {D.}~\bibnamefont
  {Gottesman}},\ }\href@noop {} {\bibinfo {title} {Stabilizer codes and quantum
  error correction}} (\bibinfo {year} {1997}),\ \Eprint
  {https://arxiv.org/abs/arXiv:quant-ph/9705052} {arXiv:quant-ph/9705052}
  \BibitemShut {NoStop}%
\bibitem [{\citenamefont {Nielsen}\ and\ \citenamefont
  {Chuang}(2000)}]{Nielsen-2000a}%
  \BibitemOpen
  \bibfield  {author} {\bibinfo {author} {\bibfnamefont {M.~A.}\ \bibnamefont
  {Nielsen}}\ and\ \bibinfo {author} {\bibfnamefont {I.~L.}\ \bibnamefont
  {Chuang}},\ }\href@noop {} {\emph {\bibinfo {title} {Quantum {C}omputation
  and {Q}uantum {I}nformation}}}\ (\bibinfo  {publisher} {Cambridge University
  Press},\ \bibinfo {year} {2000})\BibitemShut {NoStop}%
\bibitem [{\citenamefont {Fattal}\ \emph {et~al.}(2004)\citenamefont {Fattal},
  \citenamefont {Cubitt}, \citenamefont {Yamamoto}, \citenamefont {Bravyi},\
  and\ \citenamefont {Chuang}}]{Fattal-2004a}%
  \BibitemOpen
  \bibfield  {author} {\bibinfo {author} {\bibfnamefont {D.}~\bibnamefont
  {Fattal}}, \bibinfo {author} {\bibfnamefont {T.~S.}\ \bibnamefont {Cubitt}},
  \bibinfo {author} {\bibfnamefont {Y.}~\bibnamefont {Yamamoto}}, \bibinfo
  {author} {\bibfnamefont {S.}~\bibnamefont {Bravyi}},\ and\ \bibinfo {author}
  {\bibfnamefont {I.~L.}\ \bibnamefont {Chuang}},\ }\href@noop {} {\bibinfo
  {title} {Entanglement in the stabilizer formalism}} (\bibinfo {year}
  {2004}),\ \Eprint {https://arxiv.org/abs/arXiv:quant-ph/0406168}
  {arXiv:quant-ph/0406168} \BibitemShut {NoStop}%
\bibitem [{Note1()}]{Note1}%
  \BibitemOpen
  \bibinfo {note} {One feature of Pusey's model in Ref. \cite {Pusey2010} that
  can also be applied here is a slight simplification of the shared classical
  randomness. Instead of needing a trio of random bits $\lambda _e=(\protect
  \overline {X}_e,\protect \overline {Y}_e,\protect \overline {Z}_e)$ for each
  edge $e$, it suffices to have two bits $(\protect \overline {X}_e,\protect
  \overline {Z}_e)$ and then compute the third random bit as their product,
  $\protect \overline {Y}_e=\protect \overline {X}_e\protect \overline
  {Z}_e$.}\BibitemShut {Stop}%
\bibitem [{\citenamefont {Mermin}(1990)}]{Mermin-1990a}%
  \BibitemOpen
  \bibfield  {author} {\bibinfo {author} {\bibfnamefont {N.~D.}\ \bibnamefont
  {Mermin}},\ }\bibfield  {title} {\bibinfo {title} {Simple unified form for
  the major no-hidden-variables theorems},\ }\href
  {https://doi.org/10.1103/PhysRevLett.65.3373} {\bibfield  {journal} {\bibinfo
   {journal} {Phys. Rev. Lett.}\ }\textbf {\bibinfo {volume} {65}},\ \bibinfo
  {pages} {3373} (\bibinfo {year} {1990})}\BibitemShut {NoStop}%
\bibitem [{\citenamefont {Veitch}\ \emph {et~al.}(2014)\citenamefont {Veitch},
  \citenamefont {Mousavian}, \citenamefont {Gottesman},\ and\ \citenamefont
  {Emerson}}]{Veitch-2014a}%
  \BibitemOpen
  \bibfield  {author} {\bibinfo {author} {\bibfnamefont {V.}~\bibnamefont
  {Veitch}}, \bibinfo {author} {\bibfnamefont {S.~A.~H.}\ \bibnamefont
  {Mousavian}}, \bibinfo {author} {\bibfnamefont {D.}~\bibnamefont
  {Gottesman}},\ and\ \bibinfo {author} {\bibfnamefont {J.}~\bibnamefont
  {Emerson}},\ }\bibfield  {title} {\bibinfo {title} {The resource theory of
  stabilizer quantum computation},\ }\href
  {https://doi.org/10.1088/1367-2630/16/1/013009} {\bibfield  {journal}
  {\bibinfo  {journal} {New Journal of Physics}\ }\textbf {\bibinfo {volume}
  {16}},\ \bibinfo {pages} {013009} (\bibinfo {year} {2014})}\BibitemShut
  {NoStop}%
\bibitem [{Note2()}]{Note2}%
  \BibitemOpen
  \bibinfo {note} {Note, it is not possible for both {$+g_s\in \protect \text
  {stab}(|\Phi ^+_{\protect \mathcal {G}}\rangle )$} and {$-g_s\in \protect
  \text {stab}(|\Phi ^+_{\protect \mathcal {G}}\rangle )$} since {$-\protect
  \mathbb {I}\not \in \protect \text {stab}(|\Phi ^+_{\protect \mathcal
  {G}}\rangle )$}}\BibitemShut {NoStop}%
\bibitem [{\citenamefont {Cleve}\ \emph {et~al.}(2004)\citenamefont {Cleve},
  \citenamefont {Hoyer}, \citenamefont {Toner},\ and\ \citenamefont
  {Watrous}}]{Cleve-2004a}%
  \BibitemOpen
  \bibfield  {author} {\bibinfo {author} {\bibfnamefont {R.}~\bibnamefont
  {Cleve}}, \bibinfo {author} {\bibfnamefont {P.}~\bibnamefont {Hoyer}},
  \bibinfo {author} {\bibfnamefont {B.}~\bibnamefont {Toner}},\ and\ \bibinfo
  {author} {\bibfnamefont {J.}~\bibnamefont {Watrous}},\ }\bibfield  {title}
  {\bibinfo {title} {Consequences and limits of nonlocal strategies},\ }in\
  \href {https://doi.org/10.1109/CCC.2004.1313847} {\emph {\bibinfo {booktitle}
  {Proceedings. 19th IEEE Annual Conference on Computational Complexity,
  2004.}}}\ (\bibinfo {year} {2004})\ pp.\ \bibinfo {pages}
  {236--249}\BibitemShut {NoStop}%
\end{thebibliography}%

\newpage
\onecolumngrid
\section{Supplementary Material}

Here we provide more detailed proofs for some of the claims made in the main text.

\section{k-network nonlocality through post-selection}

 We begin with Proposition 1 in the main text, restated here for convenience.
\begin{prop}
Suppose that $p(a_1,\cdots,a_N)$ is a quantum $k$-network nonlocal distribution for $N$ parties.  Then there exists an $(N+r)$-party distribution $\hat{p}(a_1,\cdots,a_N, c_{1},\cdots,c_r)$ that is quantum $2$-network nonlocal and satisfies
\begin{equation}
p(a_1,\cdots,a_N)=\hat{p}(a_1,\cdots,a_N| c_{1}=0,\cdots,c_r=0).
\end{equation}
In other words, conditioned on the new parties $\msf{C}_{1},\cdots,\msf{C}_r$ having the all-zero output, the other $N$ parties reproduce the original distribution $p$.
\end{prop}

\begin{proof}
 Let $p(\mbf{b}_1,\cdots,\mbf{b}_N)$ be a $k$-network nonlocal distribution built on some quantum $k$-network $\mc{G}=(V,E)$.  Set $r=|E|$. For each $e\in E$ we introduce a new party $\mc{C}_e$ and a set of $|e|$ new vertices controlled locally by $\mc{C}_e$.  Each of these vertices is connected to a single node in the original hyperedge $e$.  In other words, we replace each hyperedge of size $|e|$ in $\mc{G}$ by $|e|$ disjoint bipartite edges, thereby generating a new $2$-graph $\mc{G}'=(V', E')$.  The bipartite entanglement distributed on $\mc{G}'$ is sufficiently large so that the original $k$-partite state $\rho_e$ from $\msf{C}_e$ to all the other vertices via teleportation.  In this new network, the parties $\msf{A}_1,\cdots,\msf{A}_N$ perform the same local measurement as they did in $\mc{G}$, while parties $\msf{C}_1,\cdots,\msf{C}_{r}$ perform the teleportation measurements obtaining outcomes $(\mbf{c}_1,\cdots,\mbf{c}_r)$.  Since each edge in $E'$ connects exactly two parties, this procedure generates a quantum $2$-network 
distribution $\hat{p}(\mbf{b}_1,\cdots,\mbf{b}_N,\mbf{c}_1,\cdots,\mbf{c}_{r})$.  In the teleportation measurement, let $0_e$ denote the outcome in which party $\msf{C}_e$ projects onto the maximally entangled state.  When all the $\msf{C}_e$ obtain this outcome, then the post-measurement state for parties $\msf{A}_1,
\cdots,\msf{A}_N$ is $\otimes_{e\in E}\rho_e$.  Hence, their conditional distribution is exactly the original,
\begin{equation}
\hat{p}(\mbf{b}_1,\cdots,\mbf{b}_N|0_1,\cdots,0_{r})=p(\mbf{b}_1,\cdots,\mbf{b}_N).
\end{equation}   

Now suppose that $\hat{p}(\mbf{b}_1,\cdots,\mbf{b}_N,\mbf{c}_1,\cdots,\mbf{c}_{r})$ is $2$-network local on graph $\mc{G}'=(V',E')$.  Let $C_e$ denote the collection of edges in $E'$ that are connected to party $\msf{C}_e$, and note that no more than $k$ distinct parties are connected to $\msf{C}_e$ since $\mc{G}$ is a $k$-network.  Then $\hat{p}(\mbf{b}_1,\cdots,\mbf{b}_N,0_1,\cdots,0_{r})$ is equal to
\begin{equation}
\sum_{\vec{\lambda}}\hat{p}(\vec{\lambda})\prod_{i=1}^N\hat{p}(\mbf{b}_i|\!\bigcup_{e'\in A_i}\!\lambda_{e'})\prod_{e=1}^{r}\hat{p}(0_{e}|\bigcup_{e'\in C_e}\lambda_{e'}),
\end{equation}
where $\hat{p}(\vec{\lambda})=\prod_{e'\in E'}p(\lambda_{e'})=\prod_{e=1}^r\prod_{e'\in C_e}p(\lambda_{e'})$. We can write the conditional probability $\hat{p}(\mbf{b}_1,\cdots,\mbf{b}_N|0_1,\cdots,0_{r})$ as
\begin{equation}
\sum_{\vec{\lambda}}\prod_{i=1}^N\hat{p}(\mbf{b}_i|\!\bigcup_{e'\in A_i}\!\lambda_{e'})\prod_{e=1}^{r}\frac{\hat{p}(0_{e}|\bigcup_{e'\in C_e}\lambda_{e'})}{\hat{p}(0_e)}\prod_{e'\in C_e}p(\lambda_{e'}),
\end{equation}
where $\hat{p}(0_1,\cdots,0_{r})=\hat{p}(0_1)\cdots \hat{p}(0_{r})$ due to the independence. By Bayes' rule we can invert the conditional probabilities of $0_{e}$,
 \begin{equation}
\frac{\hat{p}(0_{e}|\bigcup_{e'\in C_e}\lambda_{e'})}{\hat{p}(0_e)}\prod_{e'\in C_e}\hat{p}(\lambda_{e'})=\hat{p}(\textstyle\bigcup_{e'\in C_{e}}\lambda_{e'}|0_{e}).\notag
 \end{equation}
We have therefore constructed a classical model for generating $p(\mbf{b}_1,\cdots,\mbf{b}_N)$ using shared random variables of the form $\vec{\lambda}_e:=(\lambda_{e'})_{e'\in C_{e}}$ with distribution $p(\vec{\lambda}_e|0_{e})$.  Since the edges in $C_e$ are connected to no more than $k$ distinct parties in total, the variable $\vec{\lambda}_e$ is likewise shared among no more than $k$ parties.  This contradicts the assumption that $p$ is $k$-network nonlocal, and so it must be the case that $\hat{p}$ is $2$-network nonlocal.  
\end{proof} 

\medskip

\section{Clifford gates without ancilla}

The next lemma shows that to generate network nonlocality, it suffices to consider local Clifford gates without any ancilla.

\begin{lemma}
Every Clifford $k$-network correlation $p(\mbf{b}_1,\cdots,\mbf{b}_N)$ generated using local ancilla can also be generated using local classical post-processing and no local ancilla systems. 
\end{lemma}
\begin{proof}

Let $\ket{\Psi}^{A_1\cdots A_N}$ be an $N$-partite stabilizer state with stabilizer $\stab(\ket{\Psi})$.  For an arbitrary party $\msf{A}_i$, consider the effect when all other parties perform local Clifford gates with ancilla and then measure their qubits.  When all the other parties append their local ancilla systems, the overall state $\ket{\Psi}\ket{0}_{1}\cdots\ket{0}_{m-m_i}$ has stabilizer $\stab(\ket{\Psi})\cup \{Z_{1},\cdots,Z_{m-m_i}\}$.  The application of local Cliffords will then generally mix these two sets of generators, and the new state $\ket{\Psi'}$ of $t=(n+m-m_i)$ qubits will have stabilizer $\stab(\ket{\Psi'})=\langle g_1,g_2,\cdots g_t\rangle$.  Each qubit is then measured in the $Z$-basis.  

Recalling how Pauli measurements work in the stabilizer formalism \cite{Nielsen-2000a},   if $Z_k\in \stab(\ket{\Psi'})$, then its outcome is determined and we can use $Z_k$ as one of the stabilizer generators.  On the other hand, if $Z_k\not\in\stab(\ket{\Psi'})$, then we can find a set of generators for $\stab(\ket{\Psi'})$ such that $Z_k$ anti-commutes with one and only of of them.  The measurement outcome on qubit $k$ is equivalent to an unbiased coin flip, and we replace the anti-commuting element with $(-1)^{b_k}Z_k$, where $(-1)^{b_k}$ is the measurement outcome of $Z_k$ and $b_k\in\{0,1\}$.  The post-measurement state thus has an updated stabilizer, and the measurement outcome on any other qubit will either be another random coin flip or it will depend on the value $b_{k}$.

Suppose that the measurement of parties $\msf{A}_1\cdots\msf{A}_{i-1}\msf{A}_{i+1}\cdots\msf{A}_N$ obtains the specific sequence of outcomes $\mbf{b}_{\ol{i}}=(\mbf{b}_1,\cdots,\mbf{b}_{i-1},\mbf{b}_{i+1},\cdots\mbf{b}_N)$.  The post-measurement state for system $\msf{A}_i$ is a stabilizer state $\ket{\psi_{\mbf{b}_{\ol{i}}}}^{\msf{A}_i}$ with stabilizer $\langle (-1)^{d_1} g_1,\cdots,(-1)^{d_{n_i}} g_{n_i}\rangle$, in which the bit values $d_l\in\{0,1\}$ are determined by the particular sequence $\mbf{b}_{\ol{i}}$.  Crucially, the outcomes $\mbf{b}_{\ol{i}}$ of the other parties only affect the $\pm 1$ factors on the generators $g_i$ of $\stab(\ket{\psi_{\mbf{b}_{\ol{i}}}})$.  Party $\msf{A}_i$ then introduces $m_i$ local ancilla systems and performs a Clifford across the joint state $\ket{\Psi_{\ol{x}_i}}^{\msf{A}_i}\ket{0}_1\cdots\ket{0}_{m_i}$.  This will generate a transformation in the stabilizer of the form
\begin{align}
&\langle (-1)^{d_1} g_1,\cdots,(-1)^{d_{n_i}} g_{n_i}\rangle\cup\langle Z_1\cdots Z_{m_i}\rangle\notag\\
\to\; &\langle (-1)^{d_1} g'_1,\cdots,(-1)^{d_{n_i}} g'_{n_i},h_1,\cdots,h_{m_i}\rangle.
\end{align}
Party $\msf{A}_i$ them measures the $t_i=n_i+m_i$ qubits.  With the $g_l'$ and $h_l$ being fixed by the initial state $\ket{\Psi}$ and choice of local Cliffords for all the parties, each outcome sequence $\mbf{b}_i=(b_1,\cdots,b_{t_i})$ for party $\msf{A}_i$ is obtained according to some conditional distribution $p(\mbf{b}_i|\mbf{d})$.  Since the values $\mbf{d}=(d_1,\cdots,d_{n_i})$ are a function of the sequence $\mbf{b}_{\ol{i}}$, we have
\begin{equation}
p(\mbf{b}_i|\mbf{d})=p(\mbf{b}_i|\mbf{b}_{\ol{i}}).
\end{equation}

An equivalent protocol for party $\msf{A}_i$ that generates the same outcome distribution with no ancilla involves measuring each of the generators $\{g_l\}_{l=1}^{n_i}$ directly on the state $\ket{\Psi_{\mbf{b}_{\ol{i}}}}^{\msf{A}_i}$, thereby learning the sequence $\mbf{d}$.  The party then generates the outcome sequence $\mbf{b}_i$ by classical sampling according to the conditional distribution $p(\mbf{b}_i|\mbf{d})$.  Since $p(\mbf{b}_1,\cdots,\mbf{b}_n)=p(\mbf{b}_i|\mbf{b}_{\ol{i}})p(\mbf{b}_{\ol{i}})$ and $p(\mbf{b}_i|\mbf{d})=p(\mbf{b}_i|\mbf{b}_{\ol{i}})$, this new protocol reproduces the correct global distribution.  As $\msf{A}_i$ was an arbitrary party, we can perform this modified protocol for all parties thereby removing any use of ancilla systems.

\end{proof}

\medskip

\section{classical model for Clifford k-network correlations}
In the main text, we argued that to prove Theorem 1, it suffices to consider graphs $\mc{G}$ in which each node is connected to exactly one other node, the two nodes are held by different parties, and their connecting edge represents a maximally entangled state $\ket{\Phi^+}$ shared between them.  We denote the corresponding graph state as $\ket{\Phi^+_{\mc{G}}}:=\bigotimes_{e\in E}\ket{\Phi^+_e}$.

Once the state $\ket{\Phi^+_{\mc{G}}}$ is distributed, party $\msf{A}_i$ performs a local Clifford gate $U^{(i)}$ on its $n_i$ qubits.  Then for every $j=1,\cdots,n$, the observable $Z_j$ is measured on $\ket{\Psi'}=\bigotimes_{i=1}^N U^{(i)}\ket{\Phi^+_{\mc{G}}}$.  Equivalently, we can define the Pauli operators $g_j:=\prod_{i=1}^N U^{(i)\dagger} Z_j \prod_{i=1}^N U^{(i)}$ and envision these as being performed directly on the network state $\ket{\Phi^+_{\mc{G}}}$; the generated distribution $p(\mbf{b}_1,\cdots,\mbf{b}_N)$ will be the same since
\begin{align}
    \bra{\Psi'}\bigotimes_{j=1}^{n}(\mbb{I}+(\text{-}1)^{b_j}Z_j)\ket{\Psi'}=\bra{\Phi^+_{\mc{G}}}\bigotimes_{j=1}^{n}(\mbb{I}\!+\!(\text{-}1)^{b_j}g_j)\ket{\Phi^+_{\mc{G}}}.\notag
\end{align}

Our goal is to construct a classical model that generates the distribution  $p(\mbf{b}_1,\cdots,\mbf{b}_N)$ using the same network structure $\mc{G}$.  To avoid confusion, in what follows, we will use an overline to distinguish classical random variables from hermitian operators; i.e. $\ol{B},\ol{X},\ol{Y}$, etc..  The following proposition states that a collection of binary random variables can be characterized in terms of their parity checks.
\begin{proposition}
\label{Prop:joint-pmf-correlation}
Let $\{\ol{B}_1,\cdots,\ol{B}_n\}$ be a collection of random variables with $\ol{B}_i$ having value $b_i\in\mbb{Z}_2$.  For any subset of indices $S\subset[n]$, let $\ol{B}_S:=\sum_{i\in S}\ol{B}_i$ denote the parity of variables specified by $S$, where the addition is taken modulo two.  Then the joint pmf $p(b_1,\cdots,b_n)$ for the $\ol{B}_i$ is uniquely determined by the probabilities $Pr\{\ol{B}_S=0\}$ for $S\in 2^{[n]}$.
\end{proposition}

A proof of this proposition is provided below.  Note that if $p(b_1,\cdots,b_n)$ is obtained by measuring a stabilizer state, then $Pr\{\ol{B}_S=0\}=Pr\{g_S=+1\}$ where $g_S:=\prod_{i\in S}g_i$, where $\ol{B}_i$ is the random variable corresponding to measurement outcome $b_i$.  Hence, by Proposition \ref{Prop:joint-pmf-correlation}, a classical model will correctly simulate the distribution $p(b_1,\cdots,b_n)$ if it yields random variables $\{\ol{B}_1,\cdots,\ol{B}_n\}$ whose parity checks $\ol{B}_S$ have the same distribution as the corresponding product Pauli observables $g_S$.  The $\{g_S\}_{S\in 2^{[n]}}$ form a commuting group of observables, and by properties of stabilizer measurements \cite{Nielsen-2000a}, $Pr\{g_S=+1\}=1$ if $g_S\in \text{stab}(\ket{\Phi^+_{\mc{G}}})$, $Pr\{g_S=-1\}=1$ if $-g_S\in \text{stab}(\ket{\Phi^+_{\mc{G}}})$, and  $Pr\{g_S=+1\}=\frac{1}{2}$ if $\pm g_S\not\in\text{stab}(\ket{\Phi^+_{\mc{G}}})$.  

It will be important to distinguish the elements $g_S$ for which $+g_S\in \text{stab}(\ket{\Phi^+_{\mc{G}}})$ from the $g_S$ for which $-g_S\in\text{stab}(\ket{\Phi^+_{\mc{G}}})$ \footnote{Note, it is not possible for both {$+g_s\in\text{stab}(\ket{\Phi^+_{\mc{G}}})$} and {$-g_s\in\text{stab}(\ket{\Phi^+_{\mc{G}}})$} since {$-\mbb{I}\not\in\text{stab}(\ket{\Phi^+_{\mc{G}}})$}}.  Let $G$ be the subgroup of the $\{g_S\}_{S\in 2^{[n]}}$ such that $\pm g_S\in \text{stab}(\ket{\Phi^+_{\mc{G}}})$.  An independent set of generators for $G$ can always be found, and among these generators, consider the $g_S$ such that $-g_S\in \text{stab}(\ket{\Phi^+_{\mc{G}}})$.  There exists an $n$-qubit Pauli $g_0$ that under conjugation flips the overall sign on each of these from $-1$ to $+1$, while leaving the $+1$ sign of the other generators unchanged (Prop. 10.4 in \cite{Nielsen-2000a}).  Thus, $g_0 g_S g_0^\dagger\in \text{stab}(\ket{\Phi^+_{\mc{G}}})$ for all $g_S\in G$, and if $g_S\in\text{stab}(\ket{\Phi^+_{\mc{G}}})$ then $[g_S,g_0]=0$ while if $-g_S\in \text{stab}(\ket{\Phi^+_{\mc{G}}})$ then $\{g_S,g_0\}=0$.

We now describe a classical model that reproduces these correlations.  For each $g_k$, let $\phi_X(g_k)\subset[n]$ denote the set of qubit registers for which $g_k$ acts upon by applying a Pauli-$X$, and similarly for $\phi_Y(g_k)$ and $\phi_{Z}(g_k)$.  Note that all the qubits identified in these three sets are held by the same party since $g_k$ is formed by applying a local unitary.  The classical model is then the following:
\begin{enumerate}
\item[(a)] For each edge $e\in\mc{G}$, let $\lambda_e=(\ol{X}_e,\ol{Y}_e,\ol{Z}_e)$ be a trio of independent Bernoulli random variables, each being distributed uniformly over $\mbb{Z}_2$.
\item[(b)] The classical output $\ol{B}_k$ simulating the measure outcome for qubit $k$ is given by 
\begin{align}\displaystyle\ol{B}_k=&\sum_{e}|e\cap\phi_X(g_k)|\cdot\ol{X}_e+\sum_{e}|e\cap\phi_Y(g_k)|\cdot\ol{Y}_e\notag\\
&+\sum_{e}|e\cap\phi_Z(g_k)|\cdot\ol{Z}_e+\mc{I}(\{g_k,g_0\}) \;\;\mod 2,\notag
\end{align}
\end{enumerate}
where $\mc{I}(\{g_k,g_0\})=1$ if $\{g_k,g_0\}=0$ and $\mc{I}(\{g_k,g_0\})=0$ if $[g_k,g_0]=0$.

Let us verify that this model correctly reproduces the distribution $p(\mbf{b}_1,\cdots,\mbf{b}_N)$.  As argued above, it suffices to show that $Pr\{\ol{B}_S=0\}=Pr\{g_S=+1\}$ for all $S\in 2^{[n]}$.  Observe that $\text{stab}(\ket{\Phi^+_{\mc{G}}})$ is generated by the disjoint pairs $\{X_iX_j, Z_iZ_j\}$ for every $(i,j)\in E$, which is the crucial property that imposes the graph structure onto this simulation.  The proof of correctness is completed after two more technical steps.

\noindent \textbf{Claim:} $\pm g_S\in \text{stab}(\ket{\Phi^+_{\mc{G}}})$ iff for every $e\in E$ the three sums 
$\sum_{k\in S}|e\cap\phi_X(g_k)|$, $\sum_{k\in S}|e\cap\phi_Y(g_k)|$, and $\sum_{k\in S}|e\cap\phi_Z(g_k)|$ all vanish when taken modulo $2$.  
\begin{quote}
    
\textbf{Proof:}  Let $g_S=\prod_{k\in S} g_k$.  Since $\text{stab}(\ket{\Phi^+_{\mc{G}}})$ is generated by $\{X_iX_j, Z_iZ_j\}$ for $(i,j)\in E$, then $\pm g_S\in \text{stab}(\ket{\Phi^+_{\mc{G}}})$ iff for every $(i,j)\in E$ either (i) the product $X_iX_j$ appears in $g_S$, or (ii) neither $X_i$ nor $X_j$ appear (and likewise for $Y_iY_j$ and $Z_iZ_j$).  In case (i) $X_i$ and $X_j$ must both appear in an \textit{odd} number of the constituent Pauli strings $\{g_k\}_{k\in S}$; hence $\sum_{k\in S}|e\cap\phi_X(g_k)|$ will always be an even sum (and likewise for $\sum_{k\in S}|e\cap\phi_Y(g_k)|$ and $\sum_{k\in S}|e\cap\phi_Z(g_k)|$).  In case (ii), $X_i$ and $X_j$ must both appear in an \textit{even} number of the constituent Pauli strings $\{g_k\}_{k\in S}$ and the same conclusion holds (and likewise for the $Y_i$ and $Y_j$). 

\end{quote}

\noindent \textbf{Claim:} if $g_S\in\text{stab}(\ket{\Phi^+_{\mc{G}}})$, then $\ol{B}_S=\sum_{k\in S}\ol{B}_k=\sum_{k\in S}\mc{I}(\{g_k,g_0\})$; by the definition of $g_0$, this sum equals $0$ iff $g_0$ commutes with $g_S$, and so $Pr\{\ol{B}_S=0\}=Pr\{g_S=+1\}\in\{0,1\}$.  

\begin{quote}
\textbf{Proof:}  First if $g_S\in\text{stab}(\ket{\Phi^+_{\mc{G}}})$, then using the definition of $\ol{B}_k$ given above we will have $\sum_{k\in S}\ol{B}_k=\sum_{k\in S}\mc{I}(\{g_k,g_0\})$ because we have just established that $\sum_{k\in S}|e\cap\phi_X(g_k)|$ is an even sum for each edge $e$ (and likewise for $\sum_{k\in S}|e\cap\phi_Y(g_k)|$, and $\sum_{k\in S}|e\cap\phi_Z(g_k)|$).  Second, you can directly check from the definitions that $\sum_{k\in S}\mc{I}(\{g_k,g_0\})$ vanishes ($\mod 2$) if $g_0$ commutes with $g_S$, and it is $1$ if $g_0$ anticommutes with $g_S$.  This is exactly what we want because if $g_S\in\text{stab}(\ket{\Phi^+_{\mc{G}}})$ then $Pr\{g_S=0\}=1$ and our model will then likewise have that $Pr\{\ol{B}_S=0\}=1$.   In the same way, if $-g_S\in\text{stab}(\ket{\Phi^+_{\mc{G}}})$ then $Pr\{g_S=1\}=1$ and our model will likewise have that $Pr\{\ol{B}_S=1\}=1$.  (Remark. The whole purpose of the $\mc{I}$ function is to ensure that $\ol{B}_S=1$ when $-g_S\in\text{stab}(\ket{\Phi^+_{\mc{G}}})$.  If we didn't introduce the $\mc{I}$ term our model would always have $\ol{B}_S=0$ when $\pm g_S\in\text{stab}(\ket{\Phi^+_{\mc{G}}})$.  We need a way to distinguish between the $\pm$ cases, and introducing $g_0$ along with the commuting/anti-commuting distinction does the trick.)
\end{quote}

\noindent The previous claim covers the case when $g_S\in\text{stab}(\ket{\Phi^+_{\mc{G}}})$.  On the other hand, if $g_s\not\in \text{stab}(\ket{\Phi^+_{\mc{G}}})$ then $\ol{B}_S$ will be a sum of at least one independent Bernoulli variable, and so $Pr\{\ol{B}_S=0\}=Pr\{g_S=+1\}=\frac{1}{2}$.  In conclusion, by Proposition \ref{Prop:joint-pmf-correlation}, any Clifford $k$-network distribution can be reproduced by this classical model, which completes the proof of Theorem 1.

\bigskip

\noindent\textbf{Proof of Proposition \ref{Prop:joint-pmf-correlation}:}

Let $P$ and $Q$ be probability distributions over $n-$bit strings having the same probabilities for every parity check. Proceeding by induction, assume that the proposition holds for any distributions over $(n-1)-$bit strings or fewer. Observe that the parity check on any marginal of P is also in the set of parity checks of P. Then the parity checks on the marginals $P'$ of $P$ are equivalent to those $Q'$ of $Q$ obtained by taking the same marginalization. Since $P'$ and $Q'$ are over $(n-1)-$bit strings or fewer, our assumption implies that $P' = Q'$. That is, the probability distributions of the marginals are the same if $P$ and $Q$ have the same parity checks. It remains to prove that $P=Q$.

Suppose $P\neq Q$. Then there exists a $n-$bit sequence $x$ such that $P(x)\neq Q(x)$. Without loss of generality, assume that $x$ has even parity, and $P(x)<Q(x)$. Flipping a single bit of $x$ yields another sequence $y$ with odd parity. Then $P(x)+P(y)$ is the marginalization of $P$ with respect to that bit. Then, from our assumption, $P(x)+P(y) = Q(x)+Q(y)$. This implies $P(y)>Q(y)$. Continuing to flip one bit at a time, we get $P(x')<Q(x')$ for every even-parity bit string $x'$, and $P(y')>Q(y')$ for every odd-parity string $y'$. Summing all the even-parity probabilities, we get $\sum_{x'}P(x')<\sum_{x'}Q(x')$. But the sum is a parity check over all $n-$bits, which our assumption demands to be the same for $P$ and $Q$. Thus, by contradiction, we have $P=Q$ whenever $P'=Q'$ for marginals over $(n-1)-$bit strings or fewer. 

The one-bit case $P(x_1)=Q(x_1)$ trivially serves as a base for our induction step. We conclude that, for every $n\in \mathbb{N}$, the parity check probabilities $Pr\{\ol{B}_S = 0\}$, for $S\in 2^{[n]}$, uniquely determine a pmf $P(b_1,...,b_n)$, where $\ol{B}_S$ is the parity of the variables specified by $S$.

\section{Clifford networks with mixed states}

Finally, we show that network nonlocality can be realized using \textit{mixed} stabilizer states.  However, we stress that our construction represents a disguised form of bipartite Bell nonlocality, as opposed to genuine network nonlocality, since it uses Fritz's network extension \cite{Fritz-2012a} of the magic square game \cite{Mermin-1990a, Cleve-2004a}.

\begin{proposition} 
Network nonlocality can be generated using mixed stabilizer states and Clifford gates.
\end{proposition}
\begin{proof}
  In a slight abuse of notation, we consider a scenario in which Alice and Bob each hold locally four systems $(\msf{A}_1,\msf{A}_2,\msf{A}_{12},\msf{A}_{\msf{X}})$ and $(\msf{B}_1,\msf{B}_2,\msf{B}_{12},\msf{B}_{\msf{Y}})$, respectively.  Systems $\msf{A}_1$, $\msf{A}_2$, $\msf{B}_1$, and $\msf{B}_2$ are each qubits, while $\msf{A}_{12}$ and $\msf{B}_{12}$ are two-qubit systems.  Systems $\msf{A}_{\msf{X}}$ and $\msf{B}_{\msf{Y}}$ are classical registers.  In the magic square game, Alice and Bob share two Bell states $\ket{\Phi^+}^{\msf{A}_1\msf{B}_1}\ket{\Phi^+}^{\msf{A}_2\msf{B}_2}$, with $\ket{\Phi^+}=\frac{1}{\sqrt{2}}(\ket{00}+\ket{11})$.  They perform local measurements specified by the following $3\times 3 $ grid:
\[
\begin{pmatrix}
 \mbb{I}\otimes Z&Z\otimes\mbb{I}&Z\otimes Z\\
 X\otimes\mbb{I}&\mbb{I}\otimes X&X\otimes X\\
 -X\otimes Z&-Z\otimes X&Y\otimes Y
\end{pmatrix}.
\]
In the game, Bob is asked to make a measurement on systems $\msf{B}_1\msf{B}_2$ described by one of these nine observables.  Alice is asked to measure on systems $\msf{A}_1\msf{A}_2$ all the observables along one of the rows or along one of the columns.  Thus, Alice returns a triple of $\pm 1 $ values.  Notice that each row/column defines a particular ``context'' for Bob's observable.  Consequently, the game can be equivalently described by having Alice measure in one of six orthogonal bases, each one being a common eigenbasis to the observables along a row or column.  After measuring in the given basis, she can determine the triple of $\pm 1$ values to report using her measurement outcome.  For example, the observables in the second column have a common eigenbasis of $\{\ket{0}\ket{+},\ket{0}\ket{-},\ket{1}\ket{+},\ket{1}\ket{-}\}$.  When she projects in this basis, she outputs three values of $\pm 1$ based on whether the projected state is a $\pm 1$ eigenstate of $Z\otimes\mbb{I}$, $\mbb{I}\otimes X$, and $-Z\otimes X$, respectively.  The crucial observation for the protocol we describe below is that each eigenbasis defined by a row/column in the above square is invariant under local Paulis (up to overall $\pm 1$ phases). 

Consider a line network of four nodes: $\msf{X}$, $\msf{A}$, $\msf{B}$, $\msf{Y}$.  Between $\msf{A}$ and $\msf{B}$ is shared the state $\ket{\Phi^+}^{\msf{A}_1\msf{B}_1}\ket{\Phi^+}^{\msf{A}_2\msf{B}_2}$.  Let $x\in\{1,\cdots,6\}$ enumerate the three rows and three columns in the grid.  Then, between $\msf{X}$ and $\msf{A}$ is shared the mixed state
\begin{equation}
\rho^{\msf{X}\msf{A}}=\frac{1}{6}\sum_{x=1}^6\op{x}{x}^{\msf{X}}\otimes\op{x}{x}^{\msf{A}_{\msf{X}}}\otimes\op{e_x}{e_x}^{\msf{A}_{12}},
\end{equation} 
where $\ket{e_x}^{\msf{A}_{12}}$ is any one of the four basis elements belonging to the common eigenbasis of row/column $x$.  For instance in the example above, we could take $\ket{e_5}=\ket{0}\ket{+}$, with $x=5$ denoting the second column of the grid.  Likewise for Bob, let $y\in\{1,\cdots,9\}$ enumerate all the elements in the grid.  Then, between $\msf{Y}$ and $\msf{B}$ is shared the mixed state
\begin{align}
\rho^{\msf{Y}\msf{B}}=\frac{1}{9}\sum_{y=1}^9\op{y}{y}^{\msf{Y}}\otimes\op{y}{y}^{\msf{B}_{\msf{Y}}}\otimes\frac{1}{4}\left(\mbb{I}\otimes\mbb{I}+g_y\right)^{\msf{B}_{12}},
\end{align}
where $g_y$ is the Pauli observable in cell $y$ of the grid.  Alice and Bob  each perform two Bell measurements across systems $\msf{A}_{12}:\msf{A}_1\msf{A}_2$ and $\msf{B}_{12}:\msf{B}_1\msf{B}_2$, respectively.  They also measure systems $\msf{A}_{\msf{X}}$ and $\msf{B}_{\msf{Y}}$, respectively, so that they know which observables they are measuring on the quantum systems.  This information is correlated with parties $\msf{X}$ and $\msf{Y}$, respectively, since they have access to the classical registers.

For Alice, the Bell measurements will project system $\msf{A}_{12}$ onto $\msf{A}_1\msf{A}_2$ up to some Pauli errors, and likewise system $\msf{B}_{12}$ gets projected onto $\msf{B}_1\msf{B}_2$ up to Pauli errors.  In the case of Alice, as noted above, the Pauli errors will just alter which one of the four basis states gets projected.  Since Alice knows the basis $x$, she can always correct these Pauli error classically.   To see how this works, it is easiest to consider a concrete example.  %Let $x=5$ correspond to the second column in the magic square and let $\ket{e_x}=\ket{0}\ket{+}$.  
Suppose that Alice measures $x=5$ on register $\msf{A}_{\msf{X}}$ with $\ket{e_x}=\ket{0}\ket{+}$, while her Bell measurements return an $X$ error on $\msf{A}_1$ and a $Z$ error on $\msf{A}_2$.  The effective projection on systems $\msf{A}_1\msf{A}_2$ is then $(Z\otimes X)\op{e_x}{e_x}(Z\otimes X)=\op{1}{1}\otimes\op{-}{-}$, which is another one of the elements in basis $x=5$.  Hence, to correctly play the magic square game, Alice needs to only return the values corresponding to a $\ket{1}\ket{-}$ projection; i.e. she returns $(-1,-1,-1)$.  A similar correction is performed by Bob based on his outcome of the Bell measurement.  In total, Alice and Bob correctly reproduce the quantum strategy of playing the magic square game, and their choices of measurements are independently correlated with separate nodes $\msf{X}$ and $\msf{Y}$.  The generated distribution is thus quantum $2$-network nonlocal \cite{Fritz-2012a}. 
\end{proof}

\end{document}